\documentclass[a4paper,fleqn]{cas-dc}

\usepackage[numbers]{natbib}
\usepackage[official]{eurosym}
\usepackage[shortcuts]{extdash}
\usepackage{multicol}
\usepackage{multirow}

\usepackage[frozencache=true,cachedir=.]{minted}
\usepackage{algpseudocode}
\usepackage[ruled,vlined,linesnumbered]{algorithm2e}
\SetAlCapNameFnt{\small}
\SetAlCapFnt{\small}
\SetAlgorithmName{Alg.}{Alg.}{List of Algorithms}

\usepackage{amsmath,amssymb,amsfonts,amsthm, latexsym}
\usepackage{graphicx}
 \graphicspath{{./img/}}
\usepackage{textcomp}
\usepackage{xcolor}
\def\BibTeX{{\rm B\kern-.05em{\sc i\kern-.025em b}\kern-.08em
    T\kern-.1667em\lower.7ex\hbox{E}\kern-.125emX}}

\usepackage[font=scriptsize]{caption}
\usepackage{float}
\usepackage{url}
\usepackage[english]{babel}
\usepackage[autostyle]{csquotes}
\MakeOuterQuote{"}

\newtheorem{theorem}{Theorem}[section]

\newtheorem{lemma}[theorem]{Lemma}

\begin{document}
\let\WriteBookmarks\relax
\def\floatpagepagefraction{1}
\def\textpagefraction{.001}
\shorttitle{Empowering Blockchain-based Decentralised Applications Ecosystem}
\shortauthors{Nguyen Truong et~al.}

\title [mode = title]{A Blockchain-based Trust System for Decentralised Applications: When trustless needs trust}

\author[1]{Nguyen Truong}
\cormark[1]
\ead{n.truong@imperial.ac.uk}
\address[1]{Data Science Institute, South Kensington Campus, Imperial College London, London SW7 2AZ, United Kingdom}

\author[2]{Gyu Myoung Lee}
\address[2]{Department of Computer Science, Liverpool John Moores University, Liverpool L3 3AF, United Kingdom}
\ead{g.m.lee@ljmu.ac.uk}

\author[1]{Kai Sun}
\ead{k.sun@imperial.ac.uk}

\author[1]{Florian Guitton}
\ead{f.guitton@imperial.ac.uk}

\author[1, 3]{YiKe Guo}
\address[3]{Department of Computer Science, Hong Kong Baptist University, Kowloon Tong, Hong Kong}
\ead{y.guo@imperial.ac.uk}

\cortext[cor1]{Corresponding author}

\begin{abstract}
    Blockchain technology has been envisaged to commence an era of decentralised applications and services (DApps) without the need for a trusted intermediary. Such DApps open a marketplace in which services are delivered to end-users by contributors which are then incentivised by cryptocurrencies in an automated, peer-to-peer, and trustless fashion. However, blockchain, consolidated by smart contracts, only ensures on-chain data security, autonomy and integrity of the business logic execution defined in smart contracts. It cannot guarantee the quality of service of DApps, which entirely depends on the services' performance. Thus, there is a critical need for a trust system to reduce the risk of dealing with fraudulent counterparts in a blockchain network. These reasons motivate us to develop a fully decentralised trust framework deployed on top of a blockchain platform, operating along with DApps in the marketplace to demoralise deceptive entities while encouraging trustworthy ones. The trust system works as an underlying decentralised service providing a feedback mechanism for end-users and maintaining trust relationships among them in the ecosystem accordingly. We believe this research fortifies the DApps ecosystem by introducing an universal trust middleware for DApps as well as shedding light on the implementation of a decentralised trust system.
\end{abstract}

\begin{highlights}
\item Introduce a novel concept and provision of a universal decentralised trust system that can be integrated into any DApps sharing a same Blockchain platform.
\item Present a decentralised trust model with theoretical analysis, algorithms, and simulations.
\item Provide the whole agenda of the trust system development including technical solutions, implementation reference, as well as performance evaluation.
\end{highlights}


\maketitle

\section{Introduction} \label{INT}
%Preface
The turn of the last decade brought us to the disruptive Blockchain technology (BC) that provides a trusted infrastructure for enabling a variety of decentralised applications and services (DApps) without the need for an intermediary. To actualise this vision, Smart Contracts (SCs) technology is consolidated into the BC-based infrastructure: SCs are programmed to perform services' business logic, compiled into byte-code, and deployed onto a BC platform (i.e., replicated into full-nodes in the platform) so that a user can create transactions to execute the business logic implemented in the SCs in a decentralised fashion \cite{buterin2014}. This infrastructural BC platform offers some advanced features including immutability, transparency, trace-ability, and autonomy that are promising to effectively implement plentiful DApps from financial services (i.e., cryptocurrencies trading) to numerous services such as digital asset management \cite{ref03}, provenance tracking in logistics and supply-chain \cite{ref04, ref05}, and data sharing and processing in the Internet of Things (IoT) \cite{ref06, ref07}.

Indeed, various DApps have already been developed and employed into the real-world. For instance, there are over $4000$ DApps deployed on top of the Ethereum, Tron, and EOS platforms, serving about $150k$ active users daily in 2019\footnote{\url{https://cointelegraph.com/news/report-ethereum-tron-and-eos-dominated-dapp-ecosystem-in-2019}}. This is a considerable ecosystem and a huge decentralised peer-to-peer (P2P) marketplace. Although there are numerous challenges due to the limitation of the current BC technology hindering the advancement of DApps, we believe that \textit{"everything that can be decentralized, will be decentralized" - David A. Johnston}\footnote{\url{http://www.johnstonslaw.org}}. The DApps ecosystem is just in its preliminary state and will be the future of the next-generation Internet.

\subsection{Features of DApps}
%Introduction of DApps
There are different perspectives of DApps definition and system development among the cryptocurrency space. Nonetheless, mutual perceptions were pointed out that a DApp must satisfy some requirements: $(i)$ open source so that participants can audit the system, $(ii)$ application operations and data are recorded and executed in a decentralised BC (e.g., using SCs), and $(iii)$ a crypto token is used to access the service and to contribute to the operations (e.g., token reward) \cite{DJohnston, buterin2014daos}. As of these features, ideally, DApps have the ability to operate without human intervention and to be self-sustaining because the participation of stakeholders is continuously strengthening the systems.
%Reputation-based DApps
According to \textit{Vitalik Buterin}, DApps generally fall into two overlay categories, namely fully anonymous DApps and reputation-based ones \cite{buterin2014daos}. The first category is DApps which participants are essentially anonymous and the whole service business logic is autonomously executed by a series of instant atomic operations. Pure financial services such as Bitcoin are examples of this. Another example is digital assets trading DApps such as software license, data, and digitised properties in which the ownership can be impeccably transferred once a contract (defined and implemented using SCs) has been performed \cite{truongicc2018}.

The second category refers to a type of DApps which business logic requires a reputation-like mechanism to keep track of participants' activities for trust-related purposes. For instance, DApps for data storage and computation, similar to $Dropbox$ and $Amazon$ $AWS$ in the centralised space, do require to maintain reputation-like statistic record of peers for service quality and security-related purposes (e.g., anti-DDoS). This requirement of trust is irrelevant to BC technology which supposedly ensures only data security (e.g, for distributed ledgers), autonomy and integrity of the business logic execution programmed in corresponding SCs. The quality of service (QoS) of such a DApp also depends on the service itself (i.e., how well the service handles the business logic defined in the SCs and caters to customers).

\subsection{Necessity of a Trust System in DApps Ecosystem}
DApps usage always comes with token movement from end-users to service contributors as a result of an incentive scheme, which is crucial to maintaining the service. However, due to the immutable nature, it is practically impossible to revoke any transaction once it is settled onto BC. Thus, a DApp has to make sure that end-users are dealing with trustworthy counter-parties before invoking any SCs' functions that can lead to a token payment. Intuitively, end-users tend to look for an indication of \textit{'assurance'} before using any services. Indeed, a variety of DApps share the same stance on a challenge of lacking a unified decentralised framework to evaluate the trustworthiness of participants (for instance, decentralised storage and computing (similar to cloud storage like $Dropbox$ and $Amazon$ $AWS$), home-sharing (similar to $Airbnb$), car-sharing (similar to $Uber$), or a hotel distribution and reservation service (similar to $Booking.com$) backed by a BC platform). Consequently, a trust middleware that supports DApps' end-users to transact with trustworthy counterparts is of paramount importance as it penalises deceptive participants while encouraging authentic ones. As illustrated in Fig. \ref{fig1}, DApps, built upon a BC platform empowered by a decentralised trust system, naturally build up trust with clients and create a virtuous cycles that bolster the whole DApps ecosystem growth.

\begin{figure}[!ht]
\centering
	\includegraphics[width=0.75\columnwidth]{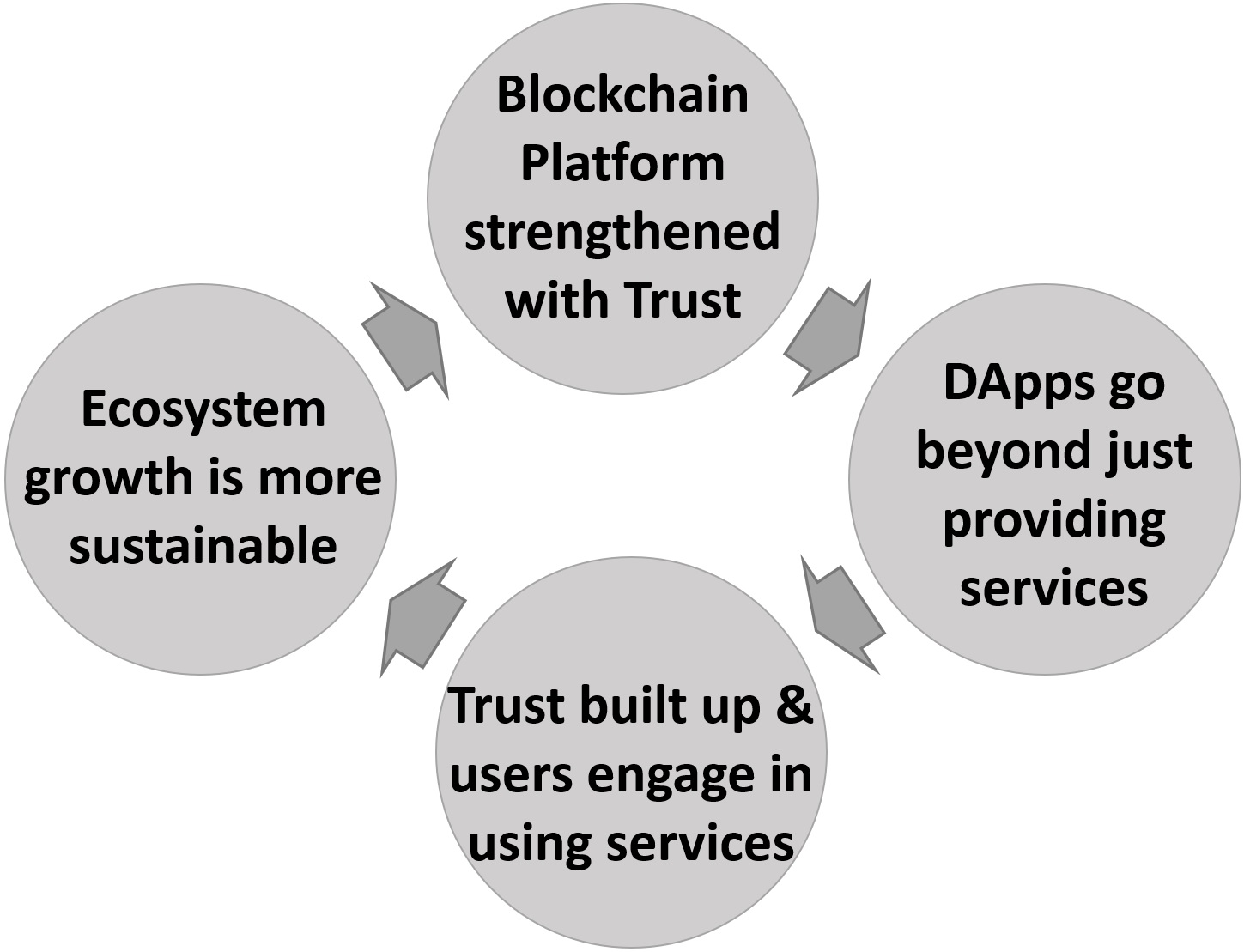}
		\caption{A BC platform strengthened with a trust system creates a virtuous cycle sustaining the DApps ecosystem growth}
\label{fig1}
\end{figure}

\subsection{Objectives and Contributions}
%Envision a Decentralised Trust Framework
Our objectives are to envision and develop a universal decentralised system that operates along with any DApps to evaluate trust relationships between entities in the ecosystem. This trust system plays as middleware between a BC platform and DApps that provides mechanisms for DApps' end-users to build up and maintain a trust relationships network among the users. Operations of the system are fully decentralised, transparent, and accessible to all of the participants which are autonomously and flawlessly executed in a trustless fashion. It is also expected to effectively prevent from reputation attacks (e.g., Sybil, White-washing, Self-promoting, and Bad\&Good-mouthing) and to dismiss masquerading hostile participants.

%Contribution
The main contributions of this paper are three-fold:
\begin{itemize}
\item Introduction to the concept and provision of a universal decentralised trust system that can be integrated into any DApps sharing a same Blockchain platform.
\item A decentralised trust model with theoretical analysis, algorithms, and simulations.
\item Providing the whole agenda of the real-world development of the system including technical solutions, implementation reference, as well as performance evaluation.
\end{itemize}

%Paper layout
The rest of the paper is organised as follows. Section II briefly brings up background and related work and presents the provision and conceptual model of a decentralised trust system. Section III describes a system design with a trust evaluation model for the proposed system. Section IV provides the algorithms and the theoretical analysis of the trust evaluation model. Section V is to discuss on the technical solutions and the implementation reference for the system development. Section VI is dedicated to the system analysis and discussion. Section VII concludes our work along with the future research directions.
\section{Decentralised Trust System Provision for DApps Ecosystem}
To craft a BC platform into a mature DApp development environment, fundamental elements must be incorporated such as an Identity Management (IdM), a name registry, a wallet, a P2P messaging for end-users, a browser, and a decentralised trust/reputation system \cite{buterin2014daos}. These elements are core built-in services of a BC-based infrastructure for DApps development.

\subsection{Related Work}
A large number of trust management mechanisms that have been proposed in various environments including social networks\cite{urena2019review}, P2P or ad-hoc networks \cite{almenarez2011trust}, and IoT \cite{yan2014survey, truong2017globecom, truong2019tifs}. Those trust models could be adapted to different scenarios including BC-related environment. However, as the emerging BC technology is in the early stage, there is limited research on trust management for DApps. Most of the related research is to develop a trust or reputation management platform leveraging the advantages of BC such as decentralisation, immutability, trace-ability, and transparency. In this respect, researchers have proposed BC-based trust mechanisms to fortify specific applications in various environments including vehicular networks and intelligent transportation systems \cite{yang2018blockchain, chen2020decentralized}, wireless sensor networks \cite{moinet2017blockchain, she2019blockchain}, or IoT \cite{debe2019iot, kochovski2019trust}. For instance, W. She \textit{et al.} in \cite{she2019blockchain} have proposed a BC-based trust model to detect malicious nodes in wireless sensor networks by implementing a voting mechanism on-chain, ensuring the trace-ability and immutability of voting information. M. Debe \textit{et al.} have developed a reputation-based trust model built on top of Ethereum platform for fog nodes in a Fog-based architectural system \cite{bonomi2012fog}. The idea is similar in that a reputation mechanism, comprising of several SCs, is implemented on top of Ethereum platform so that clients can give feedback as ratings toward a Fog node when using a service provided by such node. The reputation of a fog node is simply accumulated on-chain from users' ratings. Being executed on-chain, such ratings and reputation values are immutably recorded in a decentralised fashion, thus ensuring data integrity as well as preventing from Denial of Service (DDoS) attack.

We, instead, look at a different angle of trust in BC-based applications in which a trust system plays a complementary component of the BC platform that cooperates with DApps to empower the ecosystem built on top of the platform. We target to develop a trust system for decentralised services in a BC ecosystem (e.g., Ethereum) in which participants (clients and service providers) interact with each other on-chain in a P2P manner. Our system plays as a unified trust solution working with any DApps. Our previous research in \cite{truongicc2018} has presented an introductory concept of a unified trust system to strengthen a BC platform. However, it has come without detailed analysis, algorithm, and technical solutions for the development of the decentralised trust system. In this paper, we further explore the concept and the feasibility of a unified trust system as middleware between a BC platform and DApps, as well as provide a proof-of-concept of the decentralised trust system along with the system design, algorithms, technical solutions and implementation reference.

\subsection{High-level architecture of BC-based infrastructure and Trust System}
For a better understanding of the big picture of the whole BC-based infrastructure including the proposed trust system, we represent the high-level architecture of a full-stack IoT infrastructure by harmonising these components to the IoT and Smart Cities \& Communities reference model\footnote{\url{http://itu.int/en/ITU-T/studygroups/2017-2020/20/Pages/default.aspx}}. As can be seen in Fig. \ref{fig3}, the BC platform is located in the \textit{Service Support and Application Support} layer, which is a layer between the \textit{Application} and \textit{Network} layers in the IoT architecture. DApps is located in the \textit{Application} layer. Unlike client-server applications and services whose reputation/trust systems are separately developed, we envisage that DApps in the same ecosystem could leverage a universal trust system, which serves as a fundamental service for the BC-based infrastructure (Fig. \ref{fig3}). This trust middleware exists because DApps' end-users in an ecosystem are identified by the same IdM and a name registry, and use the same cryptocurrency (e.g., provided by a BC platform) to consume the services.

\begin{figure}[!ht]
\centering
	\includegraphics[width=\columnwidth]{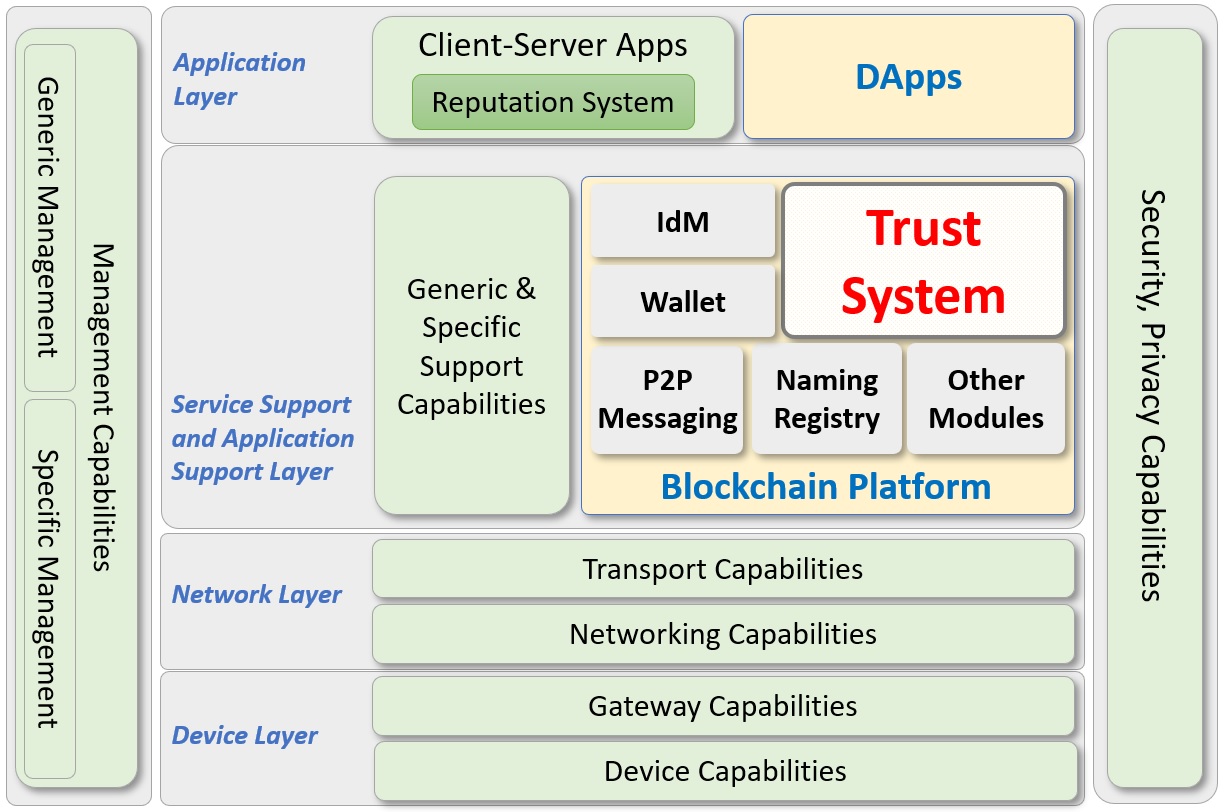}
		\caption{Functional model of a BC-based infrastructure comprising of a trust system and other elements in alignment with IoT high-level architecture.}
\label{fig3}
\end{figure}

\subsection{High-level Architecture of Trust System}
In this sub-section, fundamental elements of a decentralised trust middleware between a BC platform and DApps are described. As can be seen in Fig. \ref{fig4}, the proposed system consists of two basic components named \textit{Data Collection \& Extraction} and \textit{Trust Evaluation} that collect and aggregate necessary trust-related information and evaluate trust relationships, respectively. These two components are along with \textit{North-bound} and \textit{South-bound} APIs for providing trust-related services to DApps and for collecting data from a BC or applications and services, respectively.

\subsubsection{Trust Evaluation Mechanism}
We adopt the \textit{REK} trust model proposed in \cite{truong2017globecom, truong2019tifs} to the DApps ecosystem scenario in which both $trustors$ and $trustees$ are end-users of DApps. In the REK model, a trust relationship is evaluated by assembling three indicators called Reputation (of the trustee), Experience and Knowledge (of the trustor toward the trustee). In DApps scenarios, there is limited availability (or difficult to obtain) of off-chain information (i.e., information that is recorded outside BC) of end-users for evaluating Knowledge indicator as users' identity is normally pseudo-anonymised and challenging to link to outside world \cite{meiklejohn2013fistful}. Instead, transactions between end-users are immutably recorded (and publicly available) on-chain, which can be leveraged for Experience and Reputation evaluations. As a result, in this paper, we employ an adoption of the \textit{REK} trust evaluation model called \textit{DER} which only utilises two indicators Experience and Reputation in decentralised environment. Details of the \textit{DER} trust system is described in the next section.

\begin{figure}[!ht]
\centering
	\includegraphics[width=0.4\textwidth]{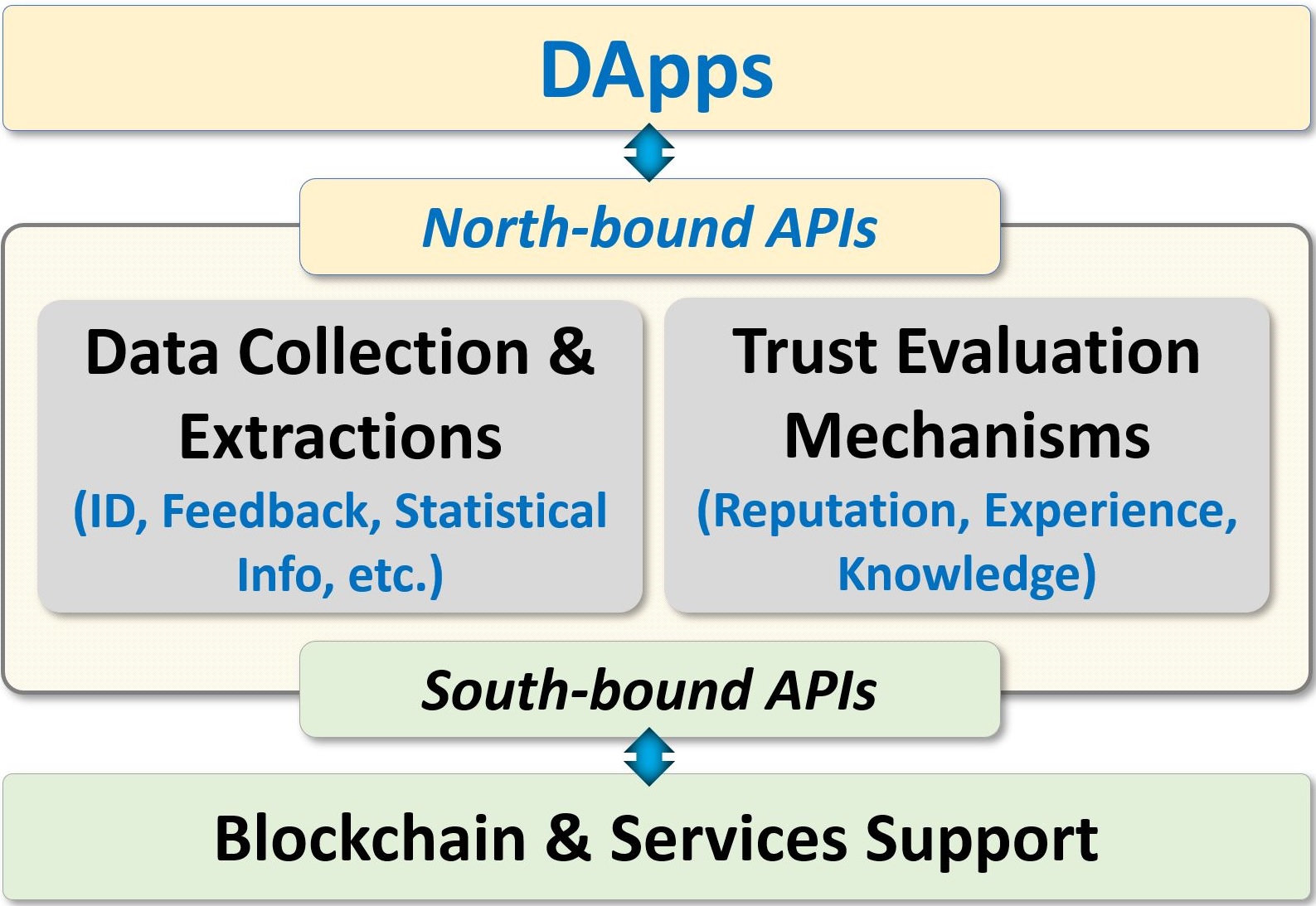}
		\caption{Conceptual model of the proposed trust system}
\label{fig4}
\end{figure}

Generally, after each transaction between entities in a DApp, the trust system enables a participant to give feedback toward its counterpart, thus establishing and updating the \textit{Experience} relationship between the two. By doing this, the trust system maintains an \textit{Experience} network among participants, which is publicly recorded on-chain. This \textit{Experience} network is autonomously updated whenever an entity gives feedback to the other. \textit{Reputations} of all participants are then calculated accordingly, following the idea of Google Rage-Rank algorithm. Finally, the trust value between two entity is calculated as composition between \textit{Experience} and \textit{Reputation}.

\subsubsection{Data Collection and Extraction}
By nature, a BC is a record of a continuous growing list of transactions among end-users which can be analysed to extract a network topology of end-user interactions. Nonetheless, further information about QoS is required to be collected and aggregated in order for the \textit{DER} trust evaluation mechanism to be performed. Therefore, a decentralised feedback mechanism associated with DApps in a BC platform is required to reflect QoS once end-users (e.g., service clients) successfully carry out transactions with their counterparts (e.g., DApp providers). This mechanism creates a \textit{distributed ledger} that logs users' feedback (toward a DApps service) along with the information about associated transactions (e.g., end-user ID ($from$ address), counterpart ID ($to$ address), and $timestamp$). Feedback can be either implicit or explicit which may or may not require human participation \cite{jawaheer2014feedback}. The trust system then extracts feedback and transactions information recorded in BCs as inputs for the \textit{DER} trust evaluation model (i.e., calculate the Experience and Reputation indicators) in order to evaluate trust relationships between any two peers in the decentralised ecosystem.
\section{System Design and DER Trust Model}

\subsection{Use-cases}
%Decentralised data storage services
For better explanation and clarification, we scrutinise the decentralised data storage services (DDS), in regard to some projects being developed and implemented in the real-world like Storj\footnote{\url{https://storj.io}}, Sia\footnote{\url{https://sia.tech}}, and Filecoin\footnote{\url{https://filecoin.io}} (built on top of the InterPlanetary File System\footnote{\url{https://ipfs.io}} (IPFS)). Decentralised storage is a promising solution to cooperate or even to take over the conventional centralised cloud storage where data is split into multiple chunks and distributed to storage nodes across a P2P network. These storage nodes, as DDS providers, are expected to reliably store the data as well as provided reasonable network bandwidth with appropriate responsiveness for data owners to retrieve their data. As a reward, such storage nodes are incentivised by crypto tokens. It is worth noting that end-users in DApps ecosystem can be both data owners (DDS clients) and storage nodes (DDS providers). The decentralised storage concept is similar to the legacy P2P file sharing such as BitTorrent\footnote{\url{https://en.wikipedia.org/wiki/BitTorrent}} but fortified with advanced cryptography and encryption mechanisms as well as incentive schemes built upon a BC platform. It is expected to solve the long-standing challenges of single-point-of-control and -failure in centralised data silos, and to bring essential control of data back to the owners whilst discharging full control of cloud server managers.

\begin{figure}[!htbp]
\centering
	\includegraphics[width=0.48\textwidth]{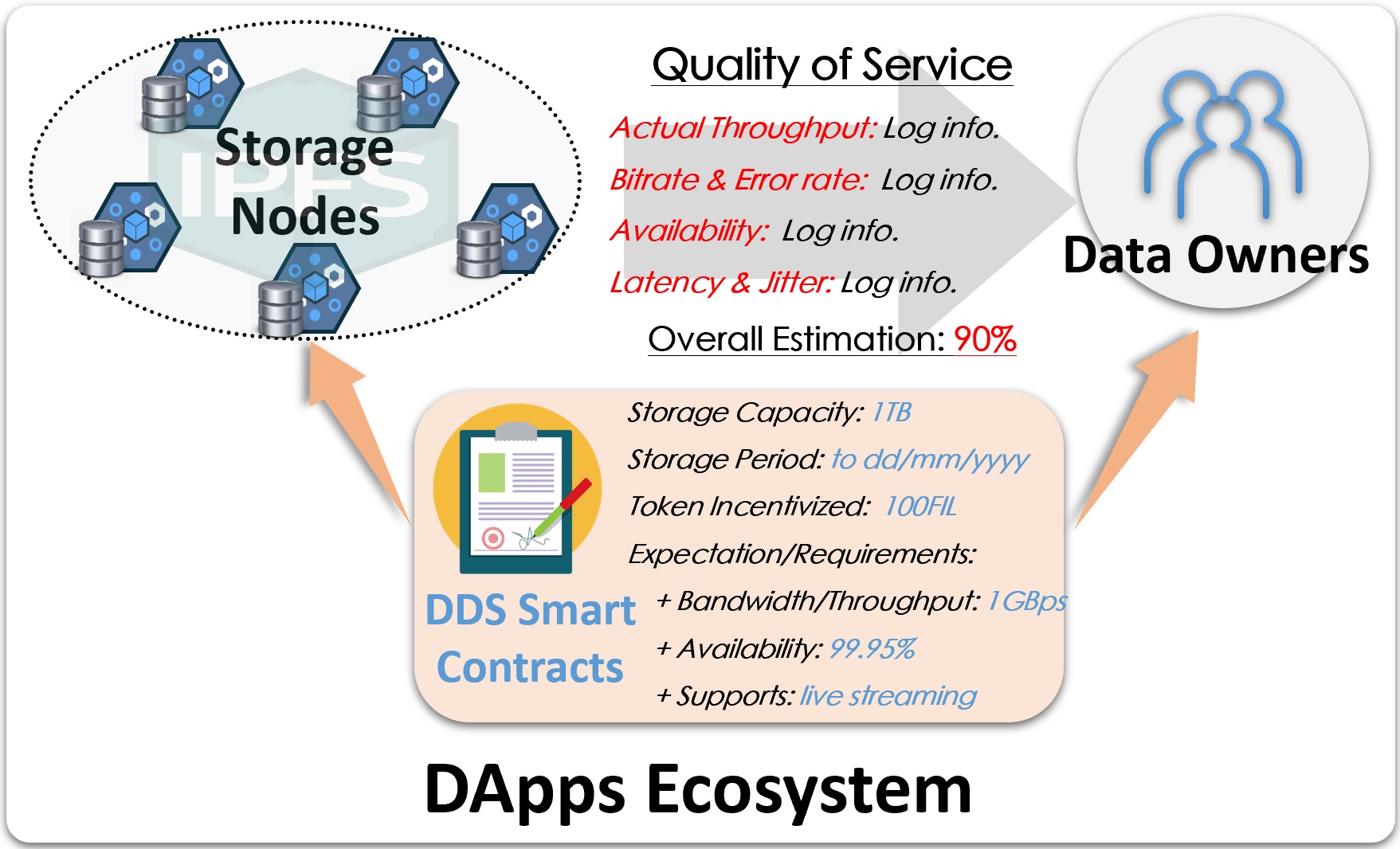}
	\caption{Decentralised storage service built on top of a BC platform that incentivizes storage nodes with crypto tokens.}
	\label{fig2}
\end{figure}

The DDS deploys necessary SCs on top of a BC platform to execute the business agreement between DDS clients (i.e., data owners) and DDS providers (i.e., storage nodes) such as \textit{storage space and period, guaranteed performance (e.g., availability, throughput, bandwidth, and latency)}, and the \textit{Incentive scheme} (i.e., \textit{Token Reward}) (Fig. \ref{fig2}). Unfortunately, such SCs are \textit{unable to ensure} the QoS of the DDS service provided by a set of storage nodes because \textit{(i)} it is impractical for the SCs to monitor and enforce the performance of the DDS providers, and \textit{(ii)} the guaranteed performance can only be measured once the SCs are already invoked. In this regard, a trust system that manages the performance history of the storage nodes and ranks them in order of trustworthiness (to provide high QoS) is of paramount importance.

\subsection{DER Trust Model}
In the proposed DER model, trust relationship between two entities is a compound of two elements: \textit{Experience} (of the trustor toward the trustee) and \textit{Reputation} (of the trustee). This section describes the mechanisms to calculate such two elements.

\subsubsection{Experience mechanism}
Experience is an asymmetric relationship from an entity to the another which is built up from previous transactions between the two. Experience is an indicator of trust \cite{truong2017globecom}. For instance, an experience (denoted as $Exp(A,B)$) is constituted from a DDS client (i.e., a data owner, denoted as $A$) to a DDS provider (i.e., a storage node, denoted as $B$) once $A$ invokes an SC to use the storage service offered by $B$. Higher $Exp(A,B)$ value represents higher degree of trust from $A$ to $B$. Essentially, $Exp(A,B)$ increases if $B$ provides high-quality storage service to $A$ (which is reflected by a feedback score $\vartheta_t$) and vice versa. It is worth noting that feedback can be provided by either clients (e.g., $A$) or an authorised third-party who is monitoring performance of service providers (e.g., $B$). Also, $Exp(A,B)$ gets decay if no transactions taken place after a period of time or a transaction is neutral (i.e., neither cooperative nor uncooperative). The amount of increase, decrease and decay depends on intensity of transactions, feedback scores $\vartheta$, and the current value of $Exp(A,B)$ which can be modelled by linear difference equations and a decay function as follows (notations are denoted in Table \ref{tb1_1}) \cite{truong2017globecom, truong2019tifs}:

\begin{table}[h]
	\centering
	\caption{NOTATIONS USED IN THE EXPERIENCE MODEL}
	\label{tb1_1}
	\begin{tabular}{c|p{0.75\columnwidth}}
			\textbf{Notation} & \textbf{Description} \\ [0.5ex]
		\hline
			$Exp_t$ & Experience value at time $t$, ${Exp_0}$ is the initial value \\
		\hline
			$min_{Exp}$ & minimum $Exp$ value, $min_{Exp} = 0$ if $Exp$ is normalised in [0,1] \\
		\hline
			$max_{Exp}$ & maximum $Exp$ value, $max_{Exp} = 1$ if $Exp$ is normalised in [0,1] \\
		\hline
			$\vartheta_t$ & Feedback score at time $t$ \\
		\hline
			$\alpha$ & Maximum increase value of $Exp$ in two consecutive transactions, $0<\alpha<max_{Exp}$ \\
		\hline
			$\beta$ & Decrease rate, $\beta>1$ \\
		\hline
			$\theta_{co}$ & Cooperative threshold for a feedback score $\vartheta_t$. A feedback is cooperative if $\vartheta_t \geq \theta_{co}$ \\
		\hline
			$\theta_{unco}$ & Uncooperative threshold for a feedback score $\vartheta_t$. A feedback is uncooperative if $\vartheta_t \leq \theta_{unco}$ \\
		\hline
			$\delta$ & Minimum Decay value ensuring any Experience relationship degenerates if it is not maintained\\
		\hline		
		  $\gamma$ & Decay rate controlling the amount of the decay\\
	\end{tabular}	\\[0.1ex]
\end{table}

\begin{itemize}
	\item \textbf{Increase model}
\end{itemize}
The current $Exp(A,B)$ (denoted as $Exp_{t-1}$) increases when there occurs a cooperative transaction (at the time $t$, indicated by the feedback score $\vartheta_t \geq \theta_{co}$) that follows the linear difference equation:

\begin{equation}
	\label{eq_1}
	Exp_t = Exp_{t-1} + \vartheta_t {\Delta}Exp_t
\end{equation}
where ${\Delta}Exp_t$ is defined as follows:
\begin{equation}
	\label{eq_2}
	{\Delta}Exp_t = \alpha(1 - \frac{Exp_{t-1}}{max_{Exp}})
\end{equation}

\begin{itemize}
	\item \textbf{Decrease model}
\end{itemize}
Similarly, $Exp(A,B)$ decreases if the transaction is uncooperative (indicated by the feedback score $\vartheta_t \leq \theta_{unco}$), following the equation:

\begin{equation}
	\label{eq_3}
	Exp_t = Max(min_{Exp}, Exp_{t-1} - \beta(1 - \vartheta_t){\Delta}Exp_t)
\end{equation}
in which ${\Delta}Exp_t$ is specified in Equation (\ref{eq_2}). The decrease rate $\beta>1$ implies that it is easier to lose the $Exp(A,B)$ value due to an uncooperative transaction than to gain it (by a cooperative transaction).

\begin{itemize}
	\item \textbf{Decay model}
\end{itemize}
$Exp(A,B)$ decays if there is no transaction after a period of time or a feedback is neutral (i.e., $\theta_{unco} < \vartheta < \theta_{co}$) and the decay rate is assumed to be inversely proportional to the strength of the experience relationship (i.e., $Exp_t$ value) \cite{roberts2009exploring}. Based on these observations, the Decay model is proposed as follows:

\begin{equation}
	\label{eq_4}
	Exp_t = Max(min_{Exp}, Exp_{t-1} - \Delta{Decay_t})
\end{equation}

\begin{equation}
	\label{eq_5}
	\Delta{Decay_t} = \delta{(1 + \gamma - \frac{Exp_{t-2}}{max_{Exp}})}
\end{equation}

\subsubsection{Reputation mechanism}
The reputation of an entity represents the overall perception of a community regarding the characteristic of the entity such as trustworthiness. In the DApps ecosystem, the reputation of an end-user $U$ (denoted as $Rep(U)$) can be calculated by aggregating $Exp(i, U)$, $\forall{i}$ are users who have already been transacted with $U$. To calculate the reputation of end-users, we utilise the model proposed in \cite{truong2017globecom, truong2019tifs} which is based on the standard PageRank \cite{originalpagerank} and the weighted PageRank \cite{xing2004weighted, weightedpagerank}.

Let $N$ be the number of end-users in the DApps ecosystem, an directed graph $G(V,E)$ is constructed in which $V$ is a set of $N$ users, $E \subseteq \{(x,y)|(x,y) \in V^2 \wedge x \ne y\}$ is set of edges representing experience relationship $E(x,y) = Exp(x,y)$. If there is no prior transaction between $(x,y)$; $E(x,y) = 0$. To enable the reputation model, $G(V,E)$ is divided into two sub-graphs: positive experience $PG(V,PE)$ in which any edge $PE(x,y) = Exp(x,y)$ satisfying $Exp(x,y)>\theta$ and negative experience $NG(V,NE)$ in which any edge $NE(x,y) = Exp(x,y)$ satisfying $Exp(x,y)<\theta$, where $\theta$ is a predefined threshold. $d$ parameter is a damping factor ($0<d<1$) introduced in standard PageRank \cite{originalpagerank}. The reputation for each sub-graph is then calculated as follows:

\begin{itemize}
	\item \textbf{Positive Reputation}
\end{itemize}
\begin{equation}
	\label{eq_6}
	Rep_{Pos}(U) = \frac{1-d}{N} + d(\sum_{\forall{i}}Rep_{Pos}(i)\times\frac{PE(i, U)}{C_{Pos}(i)})
\end{equation}
in which $C_{Pos}(i) = \sum_{\forall{j}}{PE(i,j)}$ representing the sum of all positive experience values that the end-user $i$ holds (toward other end-users).

\begin{itemize}
	\item \textbf{Negative Reputation}
\end{itemize}
\begin{equation}
	\label{eq_7}
	Rep_{Neg}(U) = \frac{1-d}{N} + d(\sum_{\forall{i}}Rep_{Neg}(i)\times\frac{1- NE(i, U)}{C_{Neg}(i)})
\end{equation}
in which $C_{Neg}(i) = \sum_{\forall{j}}{(1- NE(i,j))}$ representing the sum of all complements of negative experience values (i.e., $1 - NE(i, j)$) that the end-user $i$ holds (toward other end-users).

\begin{itemize}
	\item \textbf{Overall Reputation}
\end{itemize}
$Rep(U)$ is the aggregation of $Rep_{Pos}(U)$ and $Rep_{Neg}(U)$:
\begin{equation}
	\label{eq_8}
	Rep(U) = max(0, Rep_{Pos}(U) - Rep_{Neg}(U))
\end{equation}

\subsubsection{Trust Aggregation}
Trust relationship between trustor $A$ and trustee $B$ is a composite of $Exp(A,B)$ and $Rep(B)$:

\begin{equation}
	\label{eq_9}
	Trust(A, B) = w_1 Rep(B) + w_2 Exp(A, B)
\end{equation}
in which $w_1$ and $w_2$ are weighting factors satisfying $w_1 + w_2 = 1$. It is worth noting that any end-user once signing up for a DApp is assigned a default value at bootstrap (e.g., $\frac{1}{N}$). If $A$ and $B$ have no prior transaction then $Exp(A,B) = 0$. In this case, $w_1 = 1$ and $w_2 = 0$; thus, $Trust(A,B) = Rep(B)$. 
\section{Trust Model: Evaluation and Simulation}
This section provides detailed evaluation of the \textit{DER} trust model including model equations analysis, algorithms, and simulation of the Experience and Reputation models.

\subsection{Experience Model}
\subsubsection{Analysis}
For simplicity, $Exp$ values and feedback score $\vartheta$ are normalised to the range $(0,1)$ with $max_{Exp} = 1$, $min_{Exp} = 0$ and the initial value $0 < Exp_0 < 1$.

\begin{lemma}
	The Increase model defined in Equation \ref{eq_1} is \textit{(*)} a monotonically increasing function and \textit{(**)} \textit{asymptotic to $1$}.
\end{lemma}
 
\begin{proof}
	From Equation \ref{eq_1} and \ref{eq_2}, with $max_{Exp} = 1$, we have:
	\begin{equation}
		\label{eq_1x}
		Exp_t = Exp_{t-1} + (1 - Exp_{t-1})\vartheta_t\ \alpha
	\end{equation}

	Subtracting both sides of Equation \ref{eq_1x} from $1$:
	\begin{align}
		1 - Exp_t &= 1 - (Exp_{t-1} + (1 - Exp_{t-1})\vartheta_t\ \alpha) \nonumber \\
							&= (1 - Exp_{t-1})(1-\vartheta_t\ \alpha) \nonumber \\
							&= (1 - Exp_{t-2})(1-\vartheta_t\ \alpha)(1-\vartheta_{t-1}\ \alpha) \nonumber \\
							&= ... \nonumber \\
							&= (1 - Exp_0)\prod_{i=1}^{t} (1 - \vartheta_i\ \alpha)
		\label{eq_2x}
	\end{align}
	
	As $0 < Exp_0 < 1$, $0 < \alpha < max_{Exp} = 1$, and $0 < \vartheta_i < 1$ $\forall{i}$; from Equation \ref{eq_2x} we have $ 0 < Exp_t < 1$ $\forall{t}$. Therefore, $Exp_t$ function defined in Equation \ref{eq_1} is increasing as the increment value between $Exp_t$ and $Exp_{t-1}$ is $\vartheta_t\times{\Delta}Exp_t$ where ${\Delta}Exp_t = \alpha(1 - Exp_{t-1}) > 0$. Hence, Lemma (*) is proven.
	
	Furthermore, as Increase model is for cooperative transactions, meaning that $\vartheta_i \geq \theta_{co}; \forall{i} \in \{1,..,t\}$; from Equation \ref{eq_2x} we have:
	\begin{equation}
		\label{eq_3x}
		0 < 1 - Exp_t \leq (1 - Exp_0)(1 - \theta_{co}\ \alpha)^t
	\end{equation}

As $\theta_{co}$, $\alpha$, and $Exp_0$ are the three pre-defined parameters in the range $(0,1)$; therefore:
	\begin{equation}
		\label{eq_4x}
		\lim_{t\to\infty} (1 - Exp_0)(1 - \theta_{co}\ \alpha)^t = 0
	\end{equation}
	
Applying the Squeeze theorem on (\ref{eq_3x}) and (\ref{eq_4x}), we then have:
    \begin{equation}
		\label{eq_5x}
		\lim_{t\to\infty} (1 - Exp_t) = 0
    \end{equation}

In other word, the monotonically increasing $Exp_t$ function is \textit{asymptotic to $1$}; hence Lemma (**) is proven.
\end{proof}

As the Increase model is monotonically increasing, it is obvious that the Decrease model defined in Equation \ref{eq_3}, which is based on ${\Delta}Exp_t$ in Equation \ref{eq_2}, is decreasing. The decrements depend on the current $Exp_t$ value and the uncooperative $\vartheta_t$ feedback score. The decrease rate $\beta$ depicts the ratio of the decrements compared to the increments, which is normally greater than $1$ as the current experience $Exp_t$ is \textit{"difficult to gain but easy to loose"}.

The Decay model defined in Equation \ref{eq_4} ensures that an experience relationship gets weakened if there is no or neutral transactions after a period of time. This is because the decay value $\Delta{Decay_t}$ specified in Equation \ref{eq_5} is always $ > 0$ as $0 < Exp_{t-2} < 1$ $\forall{t \geq 2}$; and it is inversely proportional to $Exp_{t-2}$, implying that a strong relationship persists longer than a weak one.

\subsubsection{Algorithm and Simulation}
Based on the Experience model defined in Section $3.2.1$ along with the analysis, the algorithm calculates experience value $Exp(A,B)$ of entity $A$ toward entity $B$ is demonstrated in mathematical-style pseudo-code as in Algorithm \ref{alg_exp}. It is worth noting that the parameters controlling the Experience model are preset for our demonstration and should be optimised for specific scenarios.

\begin{algorithm}
    \small
	\SetKwInOut{Input}{Input}
	\SetKwInOut{Output}{Output}
	
	\Input{Current experience value $Exp_{t-1}$ \\ Previous experience value $Exp_{t-2}$ \\ Feedback score $\vartheta_t$}
	\Output{Updated experience value $Exp_t$}
	\BlankLine

	\textbf{Parameters Preset} \\
	    ${Exp_0} = 0.5$; \Comment{In case there is no prior transaction, $Exp_{t-1}$ and $Exp_{t-1}$ are set to $Exp_0$}\;
	    $min_{Exp} = 0$; $max_{Exp} = 1$; \Comment{Experience value is normalised in the range [0,1]}\;
	    $\theta_{co} = 0.7$; $\theta_{unco} = 0.5$; \\
	    $\alpha = 0.05$; $\beta = 1.6$; \\
	    $\delta = 0.005$; $\gamma = 0.005$ \\
    \BlankLine
    
    \textbf{Begin} \\
	\BlankLine
    \uIf{$\vartheta_t \geq \theta_{co}$}{ \Comment{Increase Model}\;
        $Exp_t = Exp_{t-1} + \vartheta_t \alpha (1 - \frac{Exp_{t-1}}{max_{Exp}})$ \\
        
    }
    \uElseIf{$0 < \vartheta_t \leq \theta_{unco}$}{ \Comment{Decrease Model}\;
        $Exp_t = Max(min_{Exp}, Exp_{t-1} - \beta(1 - \vartheta_t)\alpha(1 - \frac{Exp_{t-1}}{max_{Exp}})$ \\
        
    }
    \Else{\Comment{No transaction ($\vartheta_t = 0$) or neutral $\theta_{unco} < \vartheta_t < \theta_{co}$}\\
        \Comment{Decay Model}\;
        $Exp_t = Max(Exp_0, Exp_{t-1} - \delta{(1 + \gamma - \frac{Exp_{t-2}}{max_{Exp}})}$ \\
    }
    \BlankLine
    
	\textbf{Return} $Exp_t$
	\caption{Experience Calculation Algorithm}\label{alg_exp}
\end{algorithm}

\begin{figure}[!ht]
\centering
	\includegraphics[width=\columnwidth]{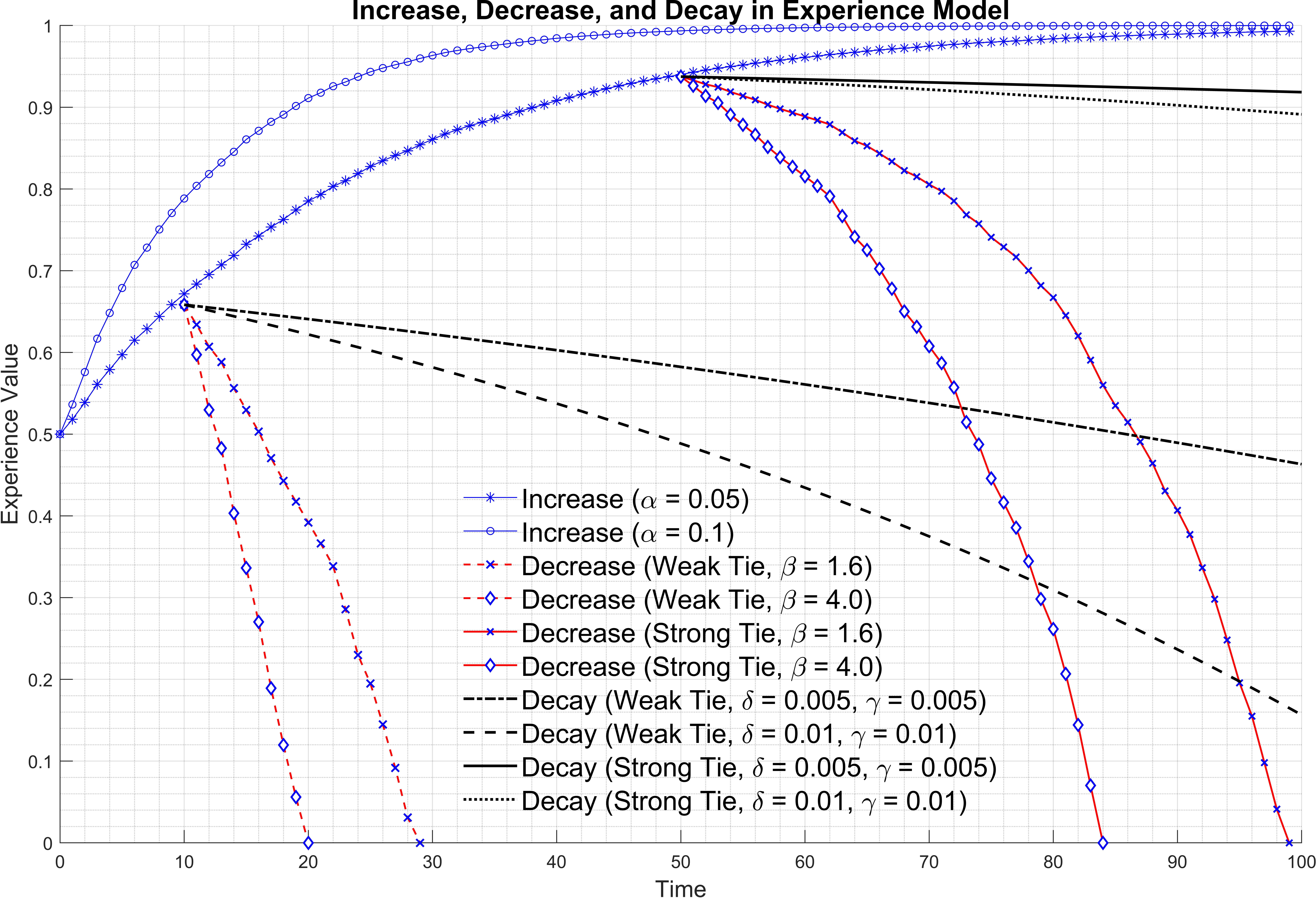}
		\caption{Increase, Decrease, and Decay in Experience relationship}
\label{fig5}
\end{figure}

For demonstration purposes, the algorithm is implemented in $Matlab$ with different controlling parameters settings. As depicted in Fig. \ref{fig5}, two sets of parameters configuration are taken into account in which the maximum increase value $\alpha$ is either $0.05$ or $0.1$, the decrease rate $\beta$ is either $1.6$ or $4.0$, and the parameter pair for the decay model ($\delta$, $\gamma$) is either ($0.005$, $0.005$) or ($0.01$, $0.01$). The initial value is preset $Exp_0 = 0.5$. As can be seen in Fig. \ref{fig5}, the results show that both increase model curves are asymptotic to $1$, which is already proven in the theoretical analysis, at different rates depending on the controlling parameter $\alpha$. The results also indicate that stronger experience relationships require more cooperative transactions to achieve. For instance, with $\alpha = 0.05$, experience value increases from $0.5$ to $0.7$ after $12$ consecutive transactions whereas it increases from $0.9$ to just $0.94$ after the same number of transactions.

The simulation results of the Decrease model show that experience relationships are prone to uncooperative transactions suggesting that a strong tie is hard to attain but easy to lose, particularly with higher decrease rate $\beta$. For instance, with $\alpha = 0.05$ and $\beta = 4.0$, it takes 50 consecutive cooperative transaction to increase the experience value from $0.5$ to $0.9$ but takes only 22 uncooperative transactions to drop from $0.9$ to $0.5$. As can also be seen from the figure, both decrease and decay models exhibit a same behaviour that a strong tie is more resistant to uncooperative transactions/decay whereas a weaker one is more susceptible. These characteristics of the experience model manifest the human social relationships, showing the practicability of the proposed model.

\subsection{Reputation Model}

\subsubsection{Analysis}
Denote $(N \times 1)$ column vectors $Rep$, $Rep_{Pos}$, and $Rep_{Neg}$ whose elements are overall reputation, positive reputation, and negative reputation of $N$ end-users in DApp ecosystem, respectively. As specified in Equation \ref{eq_6}, $Rep_{Pos}(U)$ of the user $U$ is calculated from others' positive reputations $Rep_{Pos}(i)$ $\forall{i}$ holding positive experience $PE(i,U)$ with $U$. Consequently, there would be correlations among the $N$ positive reputations, which would lead to the fact that $Rep_{Pos}$ might not exist or might be ambiguous (i.e., there exists more than one values for a user that satisfy Equation \ref{eq_6}). The same condition could happen for $Rep_{Neg}$, and for $Rep$, as a consequence.

\begin{lemma}
	The reputation vector $Rep$ \textit{exists} and is \textit{unique}.
\end{lemma}

\begin{proof}
According to Equation \ref{eq_8}, $Rep$ \textit{exists} and is \textit{unique} if both $Rep_{Pos}$ and $Rep_{Neg}$ exist and are unique.

The positive experience $N \times N$ matrix $PE$ is constituted as follows:
\begin{equation}
	\label{eq_9x}
	PE(i,j) = 
		\begin{cases}
			Exp(i,j) & \text{if } Exp(j,i) \geq \theta \\
			0 & \text{if } Exp(j,i) < \theta \\
		\end{cases}
\end{equation}

Let us constitute an $N \times N$ diagonal matrix $\mathcal{M}$ whose diagonal elements $m_i = C_{Pos}(i), \forall{i} \in \{1,..,N\}$ and a matrix $\mathcal{J}$ is a $N{\times}N$ all-ones matrix.

Based on Equation \ref{eq_6}, $Rep_{Pos}$ can be represented in matrix notation as follows:
\begin{equation}
	\label{eq_10}
	Rep_{Pos} = (\frac{1-d}{N}{\times}\mathcal{J} + d\times{PE}{\times}\mathcal{M}^{-1}){\times}Rep_{Pos}		
\end{equation}

Let us define the $A_{Pos}$ matrix as follows:
\begin{equation}
	\label{eq_11}
	A_{Pos} = \frac{1-d}{N}{\times}\mathcal{J} + d\times{PE}{\times}\mathcal{M}^{-1}
\end{equation}

Thus, Equation \ref{eq_10} can be re-written:
\begin{equation}
	\label{eq_12}
	Rep_{Pos} = A_{Pos}{\times}Rep_{Pos}		
\end{equation}

From Equation \ref{eq_12}, we can see that $Rep_{Pos}$ is the $eigenvector$ of matrix $A_{Pos}$ with the $eigenvalue = 1$. Let us define a matrix $P = A_{Pos}^T$; thus $P^T = A_{Pos}$. Therefore, Equation \ref{eq_12} can be re-written as follows:
\begin{equation}
	\label{eq_12x}
	Rep_{Pos} = P^T{\times}Rep_{Pos}
\end{equation}

Equation \ref{eq_12x} implies that $Rep_{Pos}$ is the \textit{stationary distribution} of a $Markov$ chain whose transition probability matrix is $P$. Let us constitute a discrete-time $Markov$ chain with the transition probability matrix $P = A_{Pos}^T$ consisting of $N$ states and the probability to move from state $i$ to state $j$ is $P(i,j)$. Note that $\forall{i, j} \in \{1, ..,N\}$, we have:

\begin{equation}
	\label{eq_13}
	P(i,j) = A_{Pos}^T(i,j) = A_{Pos}(j,i) = \frac{1-d}{N} + d\times\frac{PE(j,i)}{m(j)}
\end{equation}

The Markov chain can then be constructed as follows:
\begin{equation}
	\label{eq_14}
	P(i,j) = 
		\begin{cases}
			\frac{1-d}{N} + d\times\frac{PE(j,i)}{m(j)} & \text{if } Exp(j,i) \geq \theta \\
			1 - (\frac{1-d}{N} + d\times\frac{PE(j,i)}{m(j)}) & \text{if } Exp(j,i) < \theta \\
		\end{cases}
\end{equation}
where $\theta$ is the threshold to differentiate positive and negative experiences. This $Markov$ chain is a model of \textit{random surfer} with \textit{random jumps} over the experience relationships directed graph $G(V, E)$ \cite{page1999pagerank, blum2006random, chebolu2008pagerank}. The graph $G(V, E)$ is strongly connected with no dangling nodes. This is because any two nodes $(x,y)$ with no prior transaction is set $Exp(x,y) = 0$, implying that the edge weight is 0; it does not mean there is no connection. This random surfer Markov chain, apparently, is a weighted PageRank model; as a result, its \textit{stationary distribution}, $Rep_{Pos}$, exists and is \textit{unique} \cite{blum2006random, chebolu2008pagerank, haveliwala2003analytical}.

Similarly, $Rep_{Neg}$ vector \textit{exists} and is \textit{unique}. Therefore, the overall reputation vector $Rep$ \textit{exists} and is \textit{unique}.
\end{proof}

\subsubsection{Algorithm and Simulation}
As the existence and the uniqueness are proven, the reputation vector $Rep$ of $N$ end-users in DApps ecosystem can be calculated by solving the matrix equations defined in Equations \ref{eq_6}, \ref{eq_7}. The traditional algebra method to solve an $NxN$ matrix equation (e.g., Equation \ref{eq_6} or Equation \ref{eq_7}), whose the complexity is $\mathcal{O}(N^3)$, is impractical when the size of the DApp ecosystem is enormous (e.g., in millions). Instead, the reputations of the $N$ end-users can be approximately calculated with a predefined accuracy tolerance using an iterative method, which is much more efficient \cite{arasu2002pagerank, kamvar2003extrapolation}. Thus, the latter approach is utilised to solve Equations \ref{eq_6} and \ref{eq_7}, demonstrated by the following pseudo-code (Algorithm \ref{alg_rep}). As defined in Equation \ref{eq_8}, the overall reputation for $N$ end-users (i.e., $N \times 1$ column vector $Rep$) is then simply obtained by adding two vectors $Rep_{Pos}$ and $Rep_{Neg}$, which are the outputs of Algorithm \ref{alg_rep}.

\begin{algorithm}
    \small
	\SetKwInOut{Input}{Input}
	\SetKwInOut{Output}{Output}
	
	\Input{$(N \times N)$ matrix $E$ (set of edges in the directed graph $G(V,E)$ of $N$ end-users) \\ Positive reputation $N \times 1$ column vector $Rep_{Pos}$ \\ Negative reputation $N \times 1$ column vector $Rep_{Neg}$}
	
	\Output{Updated $Rep_{Pos}$ and $Rep_{Neg}$}
	\BlankLine

	\textbf{Parameters Preset} \\
	    $\mathscr{d} = 0.85$; \Comment{damping factor in standard PageRank}\\
	    $tol = 1e-5$; \Comment{Error tolerance}\\
	    $thres = 0.5$; \Comment{threshold for positive and negative experience}\\
    \BlankLine
    
    \textbf{Begin} \\
    \Comment{Elicit matrices $PE$ and $NE$ from matrix $E$}\;
    $PE = zeros(N, N)$; \Comment{initialise zero matrix for $NE$} \\
    $PE = zeros(N, N)$; \Comment{initialise zero matrix for $PE$}
    
    \For{$i \gets 1$ to $N$}{
        \For{$j \gets 1$ to $N$}{
            \uIf{$E(i,j) \geq thres)$} {
                $PE(i,j) = E(i,j)$
            }
            \ElseIf{$0 < E(i,j) < thres$} {
                $NE(i,j) = 1 - E(i,j)$
            }
        }
    }
	\BlankLine
	
    \Comment{Constitute $1 \times N$ row vectors $C_{Pos}$ and $C_{Neg}$}\;
    $C_{Pos} = zeros(1, N)$; \Comment{initialise zero vector for $C_{Pos}$} \\
    $C_{Neg} = zeros(1, N)$; \Comment{initialise zero vector for $C_{Neg}$}
    
    \For{$i \gets 1$ to $N$}{
        \For{$j \gets 1$ to $N$}{
             $C_{Pos}(1,i) = C_{Pos}(1,i) + PE(i,j)$; \\
             $C_{Neg}(1,i) = C_{Neg}(1,i) + NE(i,j)$;
        }
    }
	\BlankLine

    \Comment{Constitute transition matrices of $PE$ and $NE$}\;
        \For{$i \gets 1$ to $N$}{
        \For{$j \gets 1$ to $N$}{
            \uIf{$PE(j,i) > 0)$} {
                $A_{Pos}(i, j) = \frac{PE(j, i)}{C_{Pos}(1, j)}$; \Comment{Transition matrix for PE}
            }
            \uIf{$NE(j,i) > 0)$}{
                $A_{Neg}(i, j) = \frac{NE(j, i)}{C_{Neg}(1, j)}$; \Comment{Transition matrix for NE}
            }
        }
    }
	\BlankLine
	
    \Comment{Update $Rep_{Pos}$ and $Rep_{Neg}$ based on Equations \ref{eq_6} and \ref{eq_7}}\;

    $I = ones(N, 1)$; \Comment{create vector of all ones} \\
    $err = 1$; \Comment{Total error of the current iteration} \\
    \While{$err \geq tol$}{
        $temp_{Pos} = \mathscr{d} \times A_{Pos} \times Rep_{Pos} + \frac{(1-\mathscr{d})}{N} \times I$; \\
        $temp_{Neg} = \mathscr{d} \times A_{Neg} \times Rep_{Neg} + \frac{(1-\mathscr{d})}{N} \times I$; \\
        \BlankLine
        
        \Comment{update $err$, $\mathcal{N}(v)$ is the Euclidean norm of vector $v$}\;
        $err = \mathcal{N}(temp_{Pos} - Rep_{Pos}) + \mathcal{N}(temp_{Neg} - Rep_{Neg})$; \\
	    \BlankLine
	    
        $Rep_{Pos} = temp_{Pos}$; \Comment{update $Rep_{Pos}$ vector} \\
        $Rep_{Neg} = temp_{Neg}$; \Comment{update $Rep_{Neg}$ vector} \\
    }
    
	\textbf{Return} [$Rep_{Pos}$, $Rep_{Neg}$]
	\caption{Reputation algorithm using iterative method}\label{alg_rep}
\end{algorithm}

\begin{figure}[!ht]
\centering
	\includegraphics[width=\columnwidth]{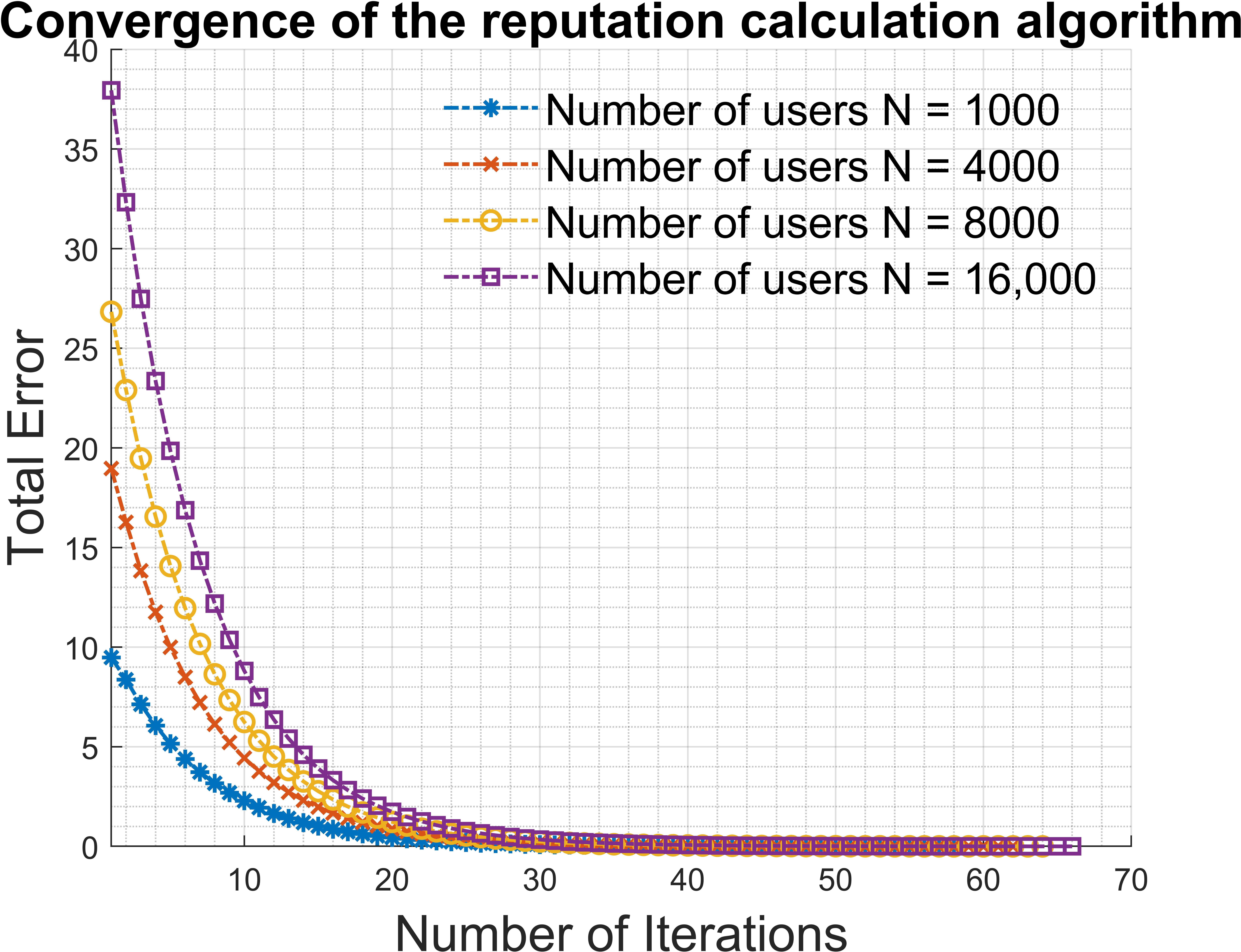}
		\caption{Convergences of the reputation algorithm using interactive method with different sizes of DApp ecosystem}
\label{fig6}
\end{figure}

The simulation of the proposed reputation calculation algorithm are conducted for different DApp ecosystem sizes (i.e., $N=1000$, $4000$, $8000$ and $16,000$) with the \textit{error tolerance} $tol = 10^{-5}$, which is accurate enough to rank $N$ end-users in the DApp ecosystem. As depicted in Algorithm \ref{alg_rep}, the total error $err$ is calculated as the Euclidean norm of the vector difference of the $Rep$ vector in two consecutive iterations. Fig. \ref{fig6} illustrates the convergence rate of the algorithm, showing the rapid reduction of the total error as more iterations are carried out. As can be seen from the figure, the algorithm converges in less than $70$ iterations (to be exact: $54$, $61$, $64$, and $66$ iterations) for four DApps ecosystem sizes $N=1000$, $4000$, $8000$ and $16,000$, respectively. These results suggests that the reputation model well scales for a huge network as the scaling factor is roughly linear in $log{N}$.

\section{Technical Solutions and Implementation}\label{imp_section}
This section provides a real-world demonstration for the proposed decentralised trust system and how a decentralised storage service interacts with it. The demonstration is carried out on top of the Ethereum permissionless BC platform in which system components, functionality, technical challenges and solutions are identified as the implementation reference for developers who wish to build a similar system. Source-code of the demonstration can be found here\footnote{\url{https://github.com/nguyentb/Decentralised\_Trust\_Eth\_IPFS.git}}. Smart Contracts source-code is in the $/packages/ethereum-core$ folder of the repository.

\subsection{System Setup}
The DDS service and the proposed decentralised trust system are implemented on top of the permissionless Ethereum platform to which fundamental elements for developing a DApp have already been deployed. For instance, in our platform setup, Ethereum $account$ and $address$ are leveraged for IdM, $Metamask$\footnote{\url{https://metamask.io/}} is for BC browser and a wallet service, and $web3/web3j$\footnote{\url{https://github.com/web3j/web3j}} are DApps APIs for interacting with Ethereum network (e.g., SCs and end-users). SCs are implemented in Solidity using Truffle suite framework\footnote{\url{https://truffleframework.com}} and deployed in an Ethereum test-net (i.e., we use several test-nets including $Ropsten$, $Kovan$ $Rinkeby$, and $Goerli$) for real-world experience. We assume that IPFS storage nodes are also clients of the DApps ecosystem (e.g., Ethereum clients in $Ropsten$, $Kovan$ or $Rinkeby$ test-net) that get incentivised once providing storage capability (e.g., IPFS storage nodes $host$ and $pin$ the hash of requested files from data owners).

The overall procedure of the setting system is illustrated in Fig. \ref{fig7}. As can be seen in the sequence diagram, a client starts to use the DDS service by making a transaction to a DDS SC (step \textbf{(1)}), which invokes \textbf{enFeedback} function in \textit{\textbf{FeEx}} SC of the trust system to grant the client permission to give feedback to the DDS nodes ((step \textbf{(3)}), \textbf{(4)})). Once getting feedback from the end-user (step \textbf{(5)}), experience relationships between the user and the DDS nodes are updated on-chain by executing \textbf{expCal} function in \textit{\textbf{FeEx}} SC (step \textbf{(6)}). On the contrary, as the reputation calculation is resource-intensive, it is impractical to implement the algorithm (i.e., Algorithm \ref{alg_rep}) on-chain; instead, only the results (i.e., reputation values of entities) are publicly recorded on-chain. This challenge can be circumvented by using Oraclize service, as demonstrated in step \textbf{(7-1)}, \textbf{(7-2)}, and \textbf{(7-3)} in Fig. \ref{fig7}. With the same reason, \textbf{Rep} SC is not invoked whenever an experience relationship is updated; instead, it is periodically self-executed - for example, for every 100 blocks.

\begin{figure}[!ht]
\centering
	\includegraphics[width=\columnwidth]{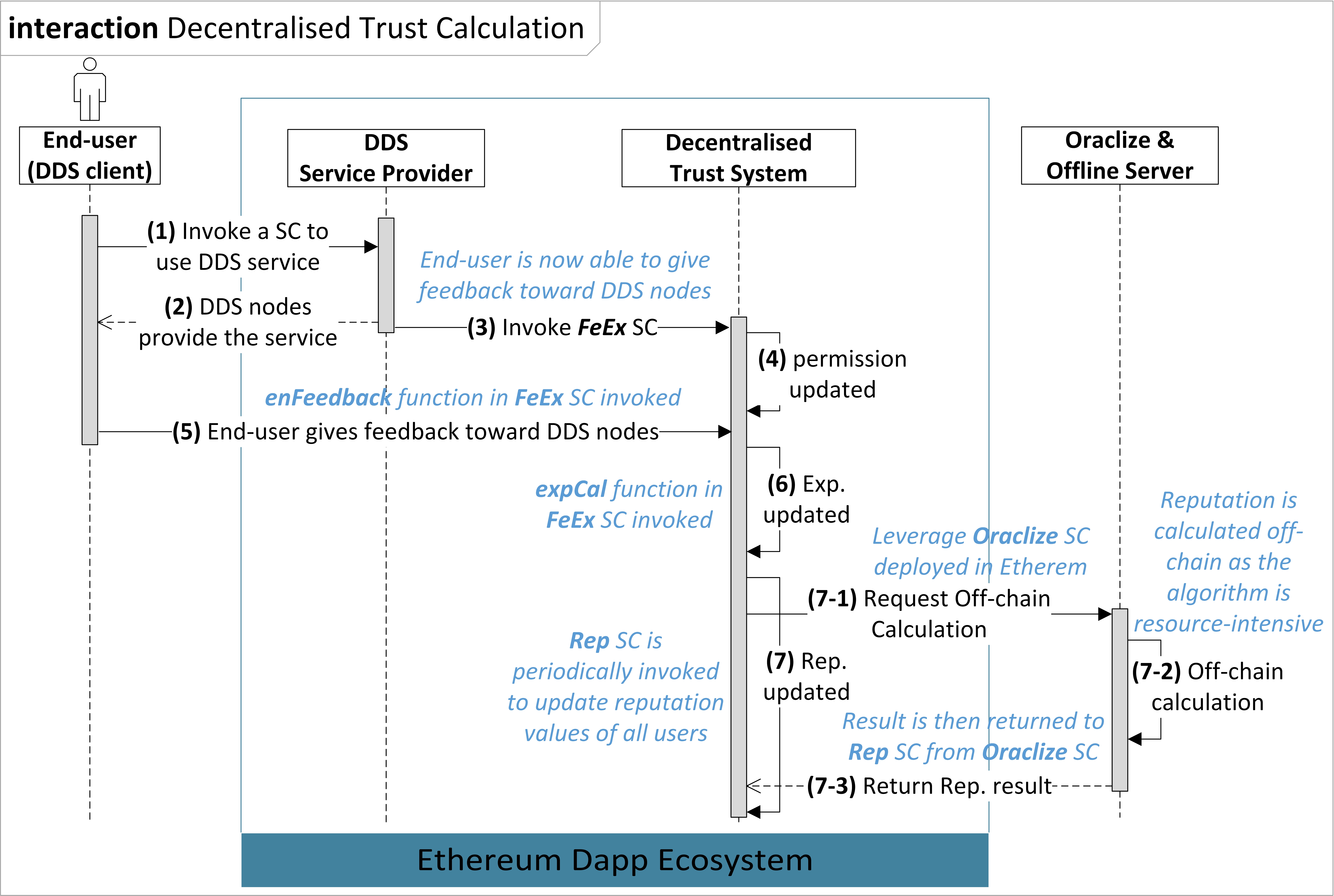}
		\caption{Sequence diagram of how the decentralised trust system is incorporated with the DDS service and how the proposed \textit{DER} trust calculation is performed}
\label{fig7}
\end{figure}

\subsection{Feedback and Experience Smart Contract}
This SC, denoted as \textit{FeEx}, contains feedback information and experience relationship of any entity $A$ (i.e., a DDS client) toward entity $B$ (an IPFS storage node) where a transaction between $A$ and $B$ has been carried out (i.e., $A$ uses the DDS service provided by $B$ depicted by step \textbf{(1)} and \textbf{(2)} in Fig. \ref{fig7}). $FeEx$ SC also provides functions for end-users to give feedback and to update experience relationships accordingly. Note that $A$ and $B$ are identified by Ethereum $address$ in the ecosystem.

\subsubsection{Ledger Data Model}
Necessary information about users' feedback and experience relationships is permanently recorded on-chain using state variables defined in \textit{FeEx} SCs. These state variables are as a public \textit{distributed ledger} comprised of the full history of \textit{state transitions} of all experience relationships between any two entities. It is convenience to obtain the latest information of any experience relationship as Ethereum supports \textit{key-value} data format and the latest state of the ledger (recording the most recent experience relationships information) can be found in the most recent block.

\textit{FeEx} SC stores a state variable, called \textit{FeExInfo}, in its contract storage in form of nested key-value pairs using Ethereum built-in $mapping$ type as follows:

\begin{small}
\begin{verbatim}
struct FeExStrut {
    uint expValue;
    uint fbScore;
    bool perFlag;
}
mapping (address=>mapping (address=>FeExStrut))
    public FeExInfo;
\end{verbatim}
\end{small}

\textit{FeExInfo} consists of information about the relationship from $A$ toward $B$, specified in \textit{FeExStrut} data structure: \textit{(ii)} $Exp(A,B)$ value, \textit{(iii)} feedback score, and \textit{(iv)} a flag indicating whether $A$ has permission to give $B$ feedback. Any parties or SCs can easily access \textit{FeExInfo} recorded on-chain to obtain desired information for their purposes.

\subsubsection{Functionality}
The \textit{FeEx} SC contains two main functions: \textit{(i)} \textit{enFeedback} enables/revokes permission of a data owner $A$ to give feedback to storage node $B$ by updating the permission flag in \textit{FeExInfo} with associated transaction ID; and \textit{(ii)} $expCal$ calculates $Exp(A,B)$ value and updates \textit{FeExInfo} whenever $A$ gives feedback to $B$. The \textit{enFeedback} function is called by by an SC of the DDS service once a transaction has been carried out (illustrated by step \textbf{(3)} in Fig. \ref{fig7}).

The $expCal$ implements the experience calculation function following Algorithm \ref{alg_exp} proposed in Section 4.1. It is worth noting that as there is no global time server synchronised among nodes in the Ethereum BC platform so that the implementation of the decay model is not straightforward. To circumvent this challenge, $expCal$ determines $time$ in Algorithm \ref{alg_exp} using block height ($block.number$ property) so that $Exp(A,B)$ decays every a number of blocks if no transaction occurred between $A$ and $B$ during the period.

\subsection{Reputation Smart Contract}
\subsubsection{Ledger Data Model}
Reputation SC, denoted as $Rep$, records positive reputation and negative reputation of all users (e.g., IPFS storage nodes) using two state variables \textit{RepPosInfo} and \textit{RepNegInfo}, respectively. The data model for the two state variables is a mapping between a user' address and a value:
\begin{small}
\begin{verbatim}
mapping (address => uint)
    public RepPosInfo;
mapping (address => uint)
    public RepNegInfo;
\end{verbatim}
\end{small}

These two state variables play the role of a public \textit{distributed ledger} permanently recording a full history of \textit{state transitions} of the positive and negative reputation of all users.

\subsubsection{Functionality}
The reputation calculation algorithm (Algorithm \ref{alg_rep}) performs matrix multiplication with numerous iterations that requires a large number of operations and local variable manipulations. Consequently, the resource-consumption and the $gas$ cost for executing this algorithm on-chain are extremely high, which is infeasible to be implemented in $Rep$ SC. To bypass this challenge, off-chain storage and calculations appear as a promising solution. The catalyst of this solution is that high-volume data and resource-intensive tasks should be stored and processed off-chain; only results of the off-chain tasks are piggybacked for on-chain ledgers and/or calculations. However, as an SC must be \textit{deterministically} executed, there might be a room for ambiguity if SC executions rely on information from off-chain sources. In addition, this practice could turn a decentralised system into a centralised one due to the dependency on an external source of information. This dilemma is known under the term: \textit{"Oracle problem"} \cite{rs03}. The following section will describe in detail how $Rep$ SC can accomplish the off-chain reputation calculation while mitigating the Oracle problem.

\subsection{Off-chain Computation for Reputation}
Oracle problem could be mitigated by leveraging a decentralised trusted provider to feed required data into SCs. For instance, Oraclize\footnote{\url{https://docs.provable.xyz/}} deploys an SC on Ethereum platform as an API for other SCs to interact with the \textit{outside world}\footnote{\url{https://github.com/provable-things/ethereum-api/blob/master/oraclizeAPI\_0.4.sol}}. The Oraclize SC works as a bearer that gets required data from an external source and delivers the data to the requested SCs in a decentralised fashion. Furthermore, to alleviate the ambiguity, it \textit{(ii)} provides \textit{authenticity proof} as an assurance for data integrity. In the implementation, we follow this Oraclize solution to calculate users' reputations off-chain.

Assume that there is already an off-chain server, called \textit{RepCalService}, that implements Algorithm \ref{alg_rep} to calculate positive and negative reputations and provides an API (e.g., REST API) to retrieve the calculation results. The implementation of this off-chain service is straightforward: it queries the Ethereum BC to obtain experience relationships stored in \textit{FeExInfo} and the current reputations values from \textit{RepPosInfo} and \textit{RepNegInfo} state variables as inputs for Algorithm \ref{alg_rep}. $Rep$ SC then periodically calls this service to update the reputation values in a decentralised fashion using Oraclize solution. The below implementation reference shows how to execute these tasks. Specifically, $Rep$ interacts with the Oraclize service by importing the Oraclize SC (i.e., \textit{provableAPI.sol}) to make a query to \textit{RepCalService} using \textit{oraclizeQuery()} function. A \textit{callback} function also needs to be implemented in order to get the results from the query and to update \textit{RepPosInfo} and \textit{RepNegInfo} accordingly.

\begin{small}
\begin{verbatim}
import "./provableAPI.sol";
contract Rep is usingProvable {
  function oraclizeQuery() {
    // make an Oraclize query to the service using URL
    oraclize_query("URL", RepCalService_API_URL);
  }
  
  function __callback(bytes32 _requestID, string _result) {
   // only Oraclize is permitted to invoke the function
   require (msg.sender == oraclize_cbAddress());
   
   // update RepPosInfo & RepNegInfo
   RepPosInfo[addr] = getRepPos(_result, addr);
   RepPosInfo[addr] = getRepNeg(_result, addr);
  }
}
\end{verbatim}
\end{small}

\subsection{Integration of DDS service and Trust System}
Supposedly, the DDS service implements some SCs for data storage business logic between data owners and storage nodes, which is out of the scope of this paper. The main focus of the paper is that once a transaction has been accomplished between a client and an IPFS storage node, the \textit{enFeedback} function in the $FeEx$ is invoked that enables the owner to give feedback to its counterpart, which will establish experience and trust relationships (step \textbf{(2)} in Fig. \ref{fig7}). For this reason, a DDS SC (i.e., the caller SC) defines an interface of $FeEx$ SC (i.e., the callee SC) and calls it with the callee's contract address as demonstrated as follows:

\begin{small}
\begin{verbatim}
contract DDS {
  function ePayment(address _storageNode,
    unit _amount, string _datahash) {
    ...
    if (success) {
      //call FeEx using deployed address scAddr
      FeEx fe = FeEx(scAddr);
      fe.enFeedback(msg.sender, _storageNode,
      string _transID);
    }
  }
}
contract FeEx {
  function enFeedback(address _owner,
    address _storageNode, string _transID);
  function expCal(address _owner, uint _fbScore,
    address _storageNode, string _transID);
}
\end{verbatim}
\end{small}

Similarly, when a data owner gives feedback toward a storage node (with value $fbScore$), DDS invokes $expCal$ function that calculates the experience relationship between the two and updates \textit{FeExInfo} accordingly. In the demonstration, feedback scores are randomly generated; however, in the real-world scenarios, a function to measure DDS QoS shall be implemented to correctly reflect the service quality. As Solidity supports interactions between SCs deployed on Ethereum platform, the proposed trust system is feasibly actualised as any DApps including DDS can be incorporated by invoking public functions or accessing trust-related information from state variables defined in the SCs of the proposed trust system.

Finally, to reinforce service quality for a client, the DDS service queries \textit{RepPosInfo}, \textit{RepNegInfo} and \textit{FeExInfo} stored at $FeEx$ and $Rep$ SCs, respectively, to receive reputation and experience values related to this client. The DDS service then aggregates this information for finalising trust values between the client and the storage nodes and provides the most trustworthy counterparts to the client.
\section{System Analysis and Discussion}
The demonstration system presented in Section \ref{imp_section} is a proof-of-concept of a universal decentralised trust system which is incorporated into a BC infrastructure as an underlying service for supporting DApps. This section investigates and discusses on the practicality, performance, and  security-related aspects of the proposed trust system.

\subsection{Feasibility and Performance Evaluation}
Practically, a variety of factors should be taken into account when deploying the trust system into real-world usages. For instance, $gas$ cost for SC execution in Ethereum Virtual Machine is high as such SCs requires high volume storage for the state variables, as well as numerous operations and local variable manipulations in $FeEx$ SC and the cost for using \textit{Oraclize} service in $Rep$ SC. This calls for further research on SC optimisation \cite{chen2017scoptimise} and better off-chain storage and calculation solutions.

As most of SCs, including $FeEx$ and $Rep$ SCs, are dedicated to performing critical tasks with minimal storage and computation, the performance of a DApp is heavily dependent on the BC platform but not the application built on top. At present, permissionless BC platforms offer limited performance in terms of both throughput and/or scalability. For instance, Bitcoin and Ethereum main-net only handle about $7$ and $15$ transactions per second\footnote{\url{https://blockchain.info/charts/n-transactions}}). In order to illustrate the real-world performance, we deploy our system to different BC platforms, i.e., Ethereum test-nets namely \textit{Ropsten}, \textit{Kovan}, \textit{Rinkeby}, and \textit{Goerli}. We carry out latency measurement of both \textit{READ} and \textit{WRITE} transactions to the ledger \textit{FeExInfo} in the \textit{FeEx} SC in the four test-nets. The results are shown in Fig. \ref{fig9}. The performance measurement script can also be found at the same repo\footnote{\url{https://github.com/nguyentb/Decentralised\_Trust\_Eth\_IPFS/tree/master/packages/performanceAnalysis}}.

\begin{figure}[!ht]
\centering
	\includegraphics[width=\columnwidth]{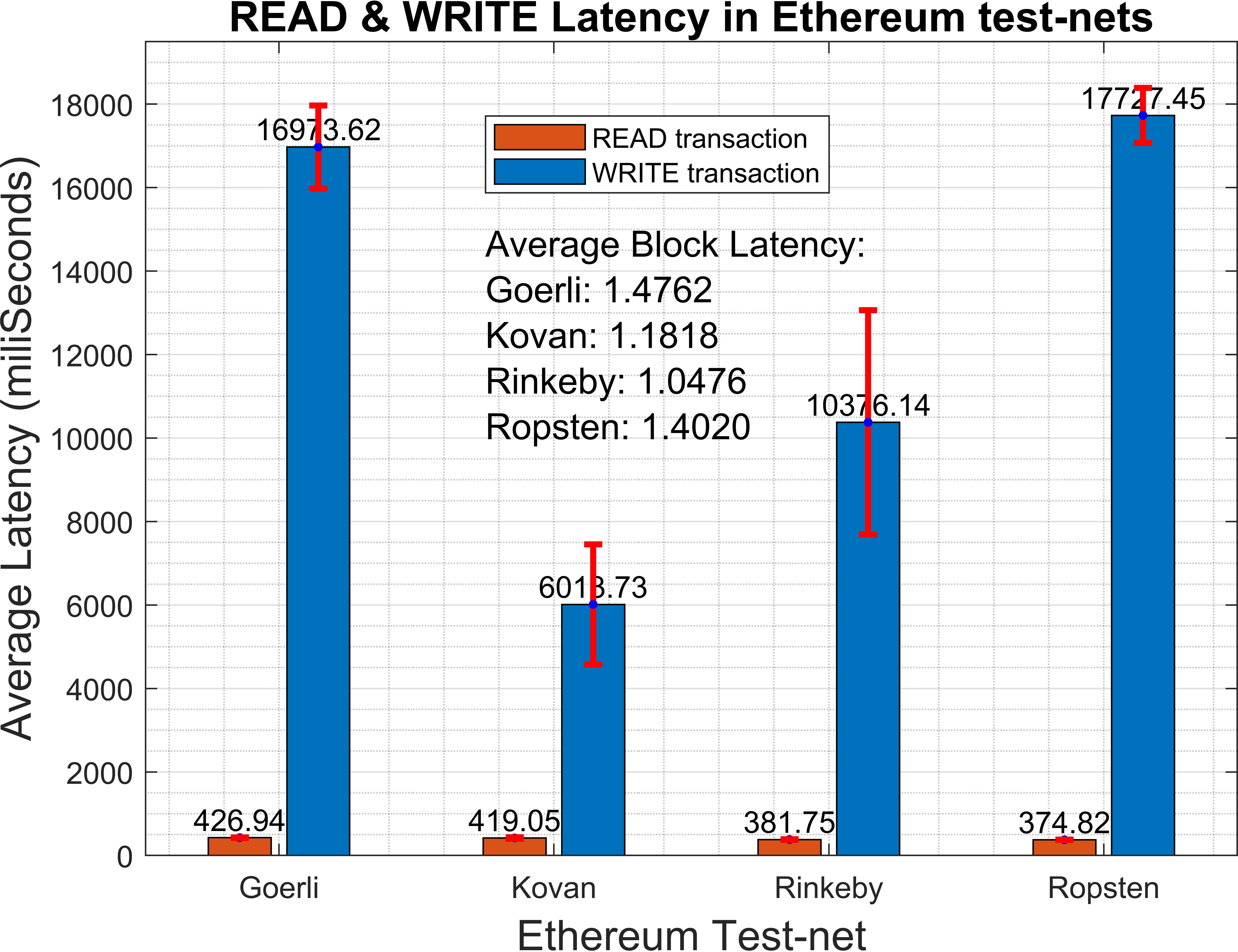}
		\caption{Latency of READ and WRITE from/to Smart Contracts in Ethereum test-nets}
\label{fig9}
\end{figure}

It is worth noting that in \textit{READ} transactions, an Ethereum platform does not perform the consensus mechanism; instead, in \textit{WRITE} transactions, consensus mechanism (i.e., Proof-of-Work (Ethash) in \textit{Ropsten}, Proof of Authority (Authority Round) in \textit{Kovan}, Proof of Authority (Clique) in both \textit{Rinkeby}, and \textit{Goerli}) is carried out as the state of the ledger is changed. In details, \textit{WRITE} transactions require further complicated processes including block formulation and mining, broadcast the mined block to peers in the network, block verification, and updating the ledger. This is why the latency of \textit{READ} transactions is much smaller than \textit{WRITE} transactions, reassured by the results in Fig. \ref{fig9}. As can be seen in the figure, the average latency of \textit{READ} transactions is roughly the same in all four test-nets at around \textit{350-420ms} with relatively small standard deviations. This indicates the consistency when querying data from the ledger. Compared to \textit{READ} transactions, the average latency in \textit{WRITE} transactions is significantly risen to $6013$, $10376$, $16973$, and $17727$ $ms$, which is $15$ to $42$ times higher, in \textit{Kovan}, \textit{Rinkeby}, \textit{Goerli}, and \textit{Ropsten}, respectively. The standard deviations, however, are different in the four test-nets: \textit{Ropsten} and \textit{Goerli} introduce considerably higher \textit{WRITE} latency compared to Kovan and Rinkeby ($2-3$ times) but \textit{WRITE} transactions are more stable as the standard deviations are small. Particularly, in \textit{Rinkeby} test-net, the standard deviation is substantially high - The latency spreads out in a wide range, from \textit{4500} to \textit{17350 ms}.

Results also show the \textit{block latency}\footnote{The number of blocks increase counted when a transaction is broadcasted to the network until it is confirmed (written in the latest block).} in \textit{WRITE} transactions in the four test-nets. In Kovan and Rinkeby, \textit{WRITE} transactions are almost appended and confirmed in the next block demonstrated by block latency is close to $1$ whereas in Goerli and Ropsten, it could take one or two more blocks before the transaction is written onto a new block. This is probably one of the reasons that the latency in \textit{Goerli} and \textit{Ropsten} is higher than in \textit{Kovan} and \textit{Rinkeby}.

Results of the system latency indicate the technical barrier on Ethereum-based system performance, which limits the usability of the proposed decentralised trust system to serve only small-scale services. Note that unlike the other test-nets, Ropsten performs \textit{Proof-of-Work} consensus mechanism, similar with the Ethereum main-net, thus, it best reproduces the Ethereum production environment. Nevertheless, besides SC optimisation for individual DApp, system performance immensely relies on an underlying BC network which requires further research on consensus mechanisms \cite{zheng2017overview}, off-chain \cite{poon2016bitcoin} and sharding solutions \cite{zamani2018rapidchain}, etc. for a better DApp ecosystem.

\subsection{System Security}
The advanced capability of BC platform plays a key role in providing a secure and trustworthy environment for DApps. Although current BC and SC technologies still pose both performance limitations and security threats, we assume that the decentralised nature of the BC ensures there is no adversary can corrupt the BC network and change the content of the ledgers as this would imply majority of the network's resources are compromised. Besides, there is no adversary who can impersonate another entity as the public-key cryptography (e.g., Elliptic Curve Digital Signature Algorithm (ECDSA) used in Ethereum) cannot be forged.

Security threats in our proposed decentralised trust system are from typical reputation-related attacks such as Self-promoting, Slandering (good/bad mouthing), and Whitewashing \cite{hoffman2009survey}. In our system, in order to be able to provide feedback, entity is required to make a transaction toward the counter-party, which costs some fee, at least transaction fee. Importantly, the proposed reputation mechanism itself can mitigate such reputation attacks. For instance, if a newly-created entity (thus its reputation value is minimal), makes a transaction, and then gives bad/good feedback toward a victim; the contribution of this feedback to the reputation value of the victim is minimal. This is because the reputation value of the victim is calculated based on both experience and reputation score of participants who transact with the victim (indicated in Equation (6) and (7) ). Obviously, if an entity is high-reputed (thus, probably not malicious) then the contribution (to one's reputation) is huge. Generally, our reputation mechanism shares the same characteristics to Page-rank algorithm in Google web-ranking engine: it is not easy to increase the ranking of a web-page by creating lots of new web-pages and link to it \cite{avrachenkov2006effect}. The nature of any feedback-based reputation systems is that it is impossible to fully prevent from such reputation attacks; however, we believe our approach can well mitigate these behaviours.
\section{Conclusion}
In this paper, we have provided a comprehensive concept, system model and design of a decentralised trust system for DApps ecosystem along with detailed analysis, algorithms, and simulations actualise the \textit{DER} trust model. Foremost, we have developed a proof-of-concept system implementing the \textit{DER} trust model on top of the Ethereum permissionless BC. The trust system is then able to incorporate with the DDS service for supporting data owners to select trustworthy storage nodes.

We have also provided technical difficulties along with prospective solutions as well as the implementation reference in the development of the proposed decentralised trust system. Existing technical barriers are also outlined which need further efforts to be successfully solved. We believe our research significantly contributes to further activities on trust-related research areas and open some future research directions to strengthen a trustworthy DApp ecosystem.

\section*{Acknowledgement}
This research was supported by the HNA Research Centre for Future Data Ecosystems at Imperial College London and the Innovative Medicines Initiative 2 IDEA-FAST project under grant agreement No 853981.

\bibliographystyle{cas-model2-names}

\bibliography{refs}

\bio{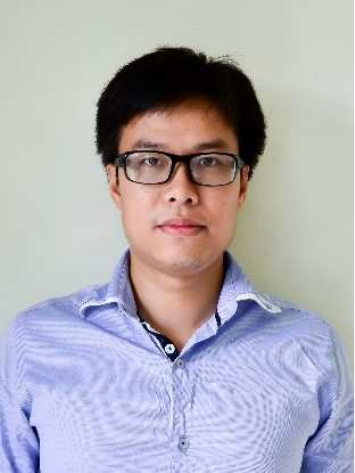}
Dr. Nguyen B.Truong is currently a Research Associate at Data Science Institute, Imperial College London, United Kingdom. He received his Ph.D, MSc, and BSc degrees from Liverpool John Moores University, United Kingdom, Pohang University of Science and Technology, Korea, and Hanoi University of Science and Technology, Vietnam in 2018, 2013, and 2008, respectively. He was a Software Engineer at DASAN Networks, a leading company on Networking Products and Services in South Korea in 2012-2015. His research interest is including, but not limited to, Data Privacy, Security, and Trust, Personal Data Management, Distributed Systems, and Blockchain.
\endbio

\bio{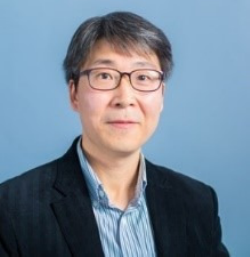}
Dr. Gyu Myoung Lee received his BS degree from Hong Ik University and MS, and PhD degrees from the Korea Advanced Institute of Science and Technology (KAIST), Korea, in 1999, 2000 and 2007, respectively. He is a Professor at Department of Computer Science, Liverpool John Moores University, UK. He is also with KAIST as an adjunct professor. His research interests include Future Networks, IoT, and multimedia services. He has actively contributed to standardization in ITU-T as a Rapporteur, oneM2M and IETF. He is chair of the ITU-T Focus Group on data processing and management to support IoT and Smart Cities \& Communities.
\endbio

\bio{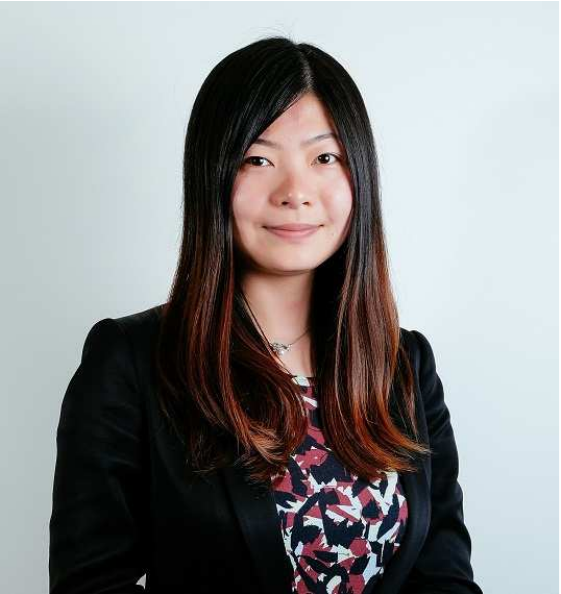}
Dr. Kai Sun is the Operation Manager of the Data Science Institute at Imperial College London. She received the MSc degree and the Ph.D degree in Computing from Imperial College London, in 2010 and 2014, respectively. From 2014 to 2017, she was a Research Associate at the Data Science Institute at Imperial College London, working on EU IMI projects including U-BIOPRED and eTRIKS, responsible for translational data management and analysis. She was the manager of the HNA Centre of Future Data Ecosystem in 2017-2018. Her research interests include translational research management, network analysis and decentralised systems.
\endbio

\bio{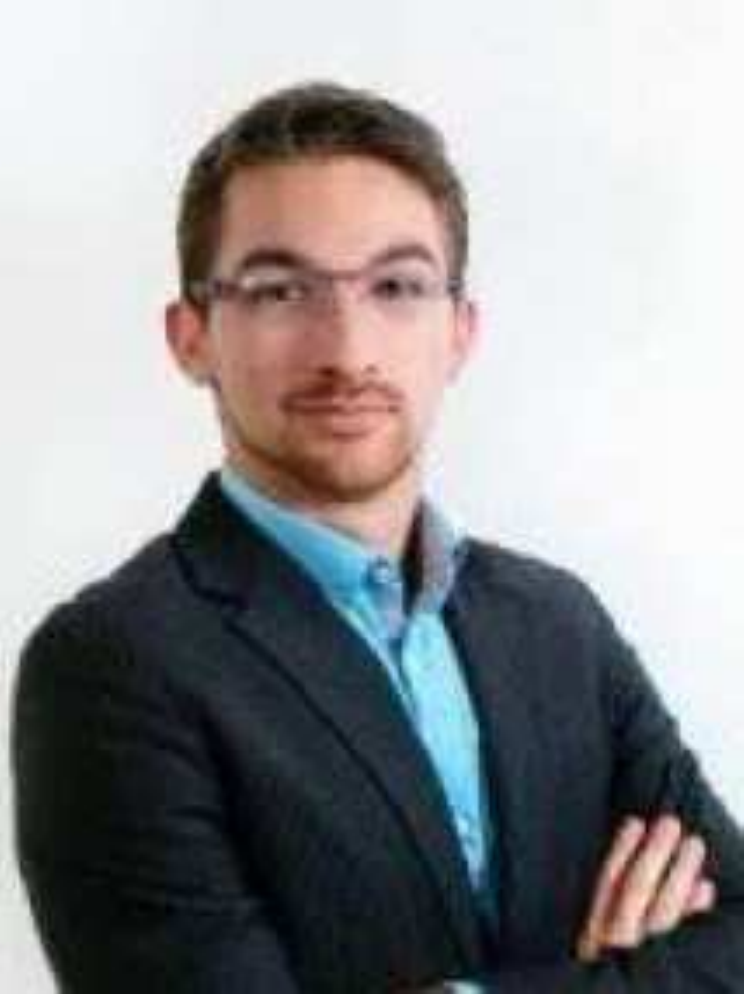}
Mr. Florian Guitton received a BSc in Software Engineering from Epitech (France) in 2011 and a MSc in Advanced Computing from the University of Kent (United Kingdom) in 2012. In 2012 he joined the Discovery Sciences Group at Imperial College London where he became Research Assistant working on iHealth, eTRIKS and IDEA-FAST EU programs. He is currently a PhD candidate at Data Science Institute, Imperial College London working on distributed data collection and analysis pipeline in mixed-security environments with the angle of optimising user facing experiences.
\endbio

\bio{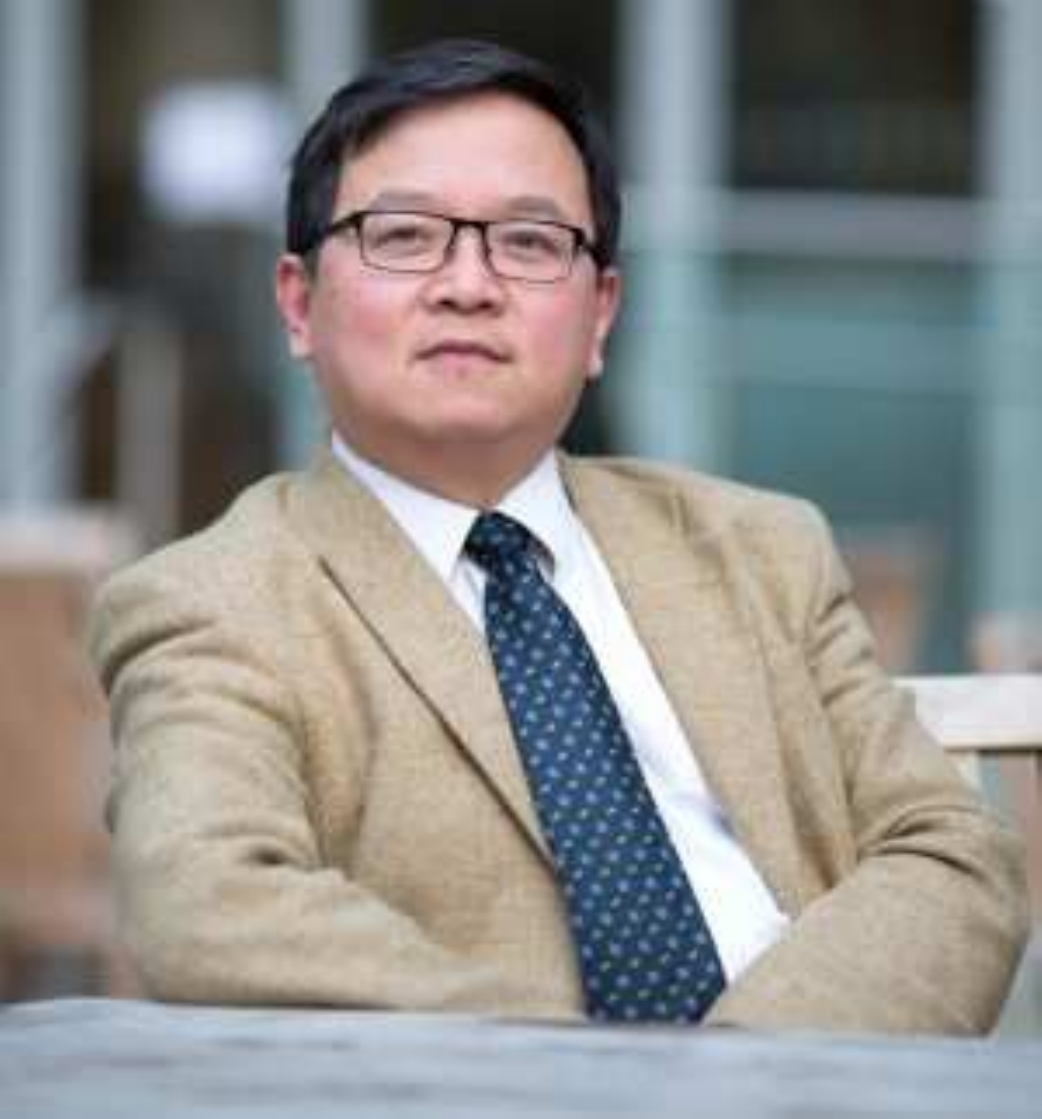}
Dr. Yike Guo (FREng, MAE) is the director of the Data Science Institute at Imperial College London and the Vice-President (Research and Development) of Hong Kong Baptist University. He received the BSc degree in Computing Science from Tsinghua University, China, in 1985 and received the Ph.D in Computational Logic from Imperial College London in 1993. He is a Professor of Computing Science in the Department of Computing at Imperial College London since 2002. He is a fellow of the Royal Academy of Engineering and a member of the Academia Europaea. His research interests are in the areas of data mining for large-scale scientific applications including distributed data mining methods, machine learning and informatics systems.
\endbio

\end{document}